[1994/12/01]% LaTeX date must December 1994 or later
\documentclass[screen,printscheme]{aomart}
\usepackage{amsfonts,bbm,amssymb,amsmath,slashed,mathrsfs}
\HSelect[red]{Editor}
\ECSelect[red]{Editor}

\fancypagestyle{firstpage}{%
  \fancyhf{}%
  \if@aom@manuscript@mode
    \lhead{\begin{picture}(0,0)%
        \put(-21,-25){\usebox{\@aom@linecount}}%
      \end{picture}}
  \fi
  \chead{}%
    \cfoot{\footnotesize\thepage}}%
\doinumber{}

\chardef\bslash=`\\ % p. 424, TeXbook
\newtheorem[{}\it]{theorem}{Theorem}[section]
\newtheorem{theoremsection}{Theorem}[section]
\newtheorem{corollarysection}[theoremsection]{Corollary}
\newtheorem{corollary}[theorem]{Corollary}
\newtheorem{lemma}[theorem]{Lemma}

\theoremstyle{definition}
\newtheorem{definition}{Definition}[section]
\newtheorem{remark}{Remark}[section]
\newtheorem{assumptions}{{\it Assumption}}[section]
\newtheorem*[{}\it]{notation}{Notation}

\def\QEDopen{{\setlength{\fboxsep}{0pt}\setlength{\fboxrule}{0.2pt}\fbox{\rule[0pt]{0pt}{1.3ex}\rule[0pt]{1.3ex}{0pt}}}}
\def\QED{\QEDopen}

\def\Q.E.D{\hfill\QED}

\newcommand{\eval}[2][\right]{\relax
  \ifx#1\right\relax \left.\fi#2#1\rvert}

\title[$\mathcal{NP}\ne{\rm co}\mathcal{NP}$]{Simulating Polynomial-Time Nondeterministic Turing Machines via Nondeterministic Turing Machines}
\author[T. Lin]{Tianrong Lin}
\address{Hakka University\\ China}
\address{Shang-Hang City, Fujian Province, China}
\givenname{Tianrong}
\surname{Lin}
\orcid{0000-0002-1187-2395}
\copyrightnote{}
\subject{primary}{matsc2020}{68Q15}
\subject{secondary}{matsc2020}{03F20}

\accepted{\formatdate{2006-06-02}}
\published{\formatdate{2008-04-02}}
\publishedonline{\formatdate{2007-11-12}}

\volumenumber{160}
\issuenumber{1}
\publicationyear{2008}
\papernumber{12}
\startpage{1}
\endpage{}

\proposed{  }
\seconded{  }
\corresponding{  }
\version{2.1}

\mrnumber{MR1154181}
\zblnumber{0774.14039}
\arxivnumber{1234.567890}
\begin{document}

\begin{abstract}

We prove in this paper that there exists a language $L_s$ accepted by some nondeterministic Turing machine that runs within time $O(n^k)$ for any positive integer $k\in\mathbb{N}_1$ but not accepted by any ${\rm co}\mathcal{NP}$ machines. We further show that $L_s$ is in $\mathcal{NP}$, thereby proving the groundbreaking result that $$\mathcal{NP}\neq{\rm co}\mathcal{NP}. $$
  
The main techniques used in this paper are simulation together with the novel techniques developed in the author's recent work. Our main result has profound implications, such as $\mathcal{P}\neq\mathcal{NP}$. Furthermore, if there exists some oracle $A$ such that $\mathcal{P}^A\ne\mathcal{NP}^A={\rm co}\mathcal{NP}^A$, we explore the underlying reasons and show that, under this condition and some reasonable assumptions, the set of all ${\rm co}\mathcal{NP}^A$ machines is not enumerable. This implies that simulation techniques cannot be applied to the first part of the separation of $\mathcal{NP}^A$ from ${\rm co}\mathcal{NP}^A$. Finally, a lower bounds result for Frege proof systems is presented (i.e., no Frege proof systems can be polynomially bounded).

\end{abstract}

\maketitle
\tableofcontents

\vskip 0.3 cm
\section{Introduction}
\label{sec:introduction}
\vskip 0.3 cm

Let $\mathcal{NP}$ be the set of decision problems solvable in polynomial time by a nondeterministic Turing machine, i.e.,

$$\aligned\mathcal{NP}\overset{\rm def}{=}&\big\{L\,:\,\text{there exists some polynomial-time nondeterministic Turing machine}\\&\qquad\,\,\text{that accepts the language $L$}\big\}. \endaligned$$

Let ${\rm co}\mathcal{NP}$ be the complement of $\mathcal{NP}$, i.e., $${\rm co}\mathcal{NP}=\left\{\overline{L}\,:\,L\in\mathcal{NP}\right\}, $$ where $\overline{L}$ is the set-theoretic complement of $L$ (i.e., assume $L\subseteq\{0,1\}^*$, then $\overline{L}=\{0,1\}^*\setminus L$). The outstanding and long-standing open problem $\mathcal{NP}\overset{?}{=}{\rm co}\mathcal{NP}$ asks whether $\mathcal{NP}$ equals ${\rm co}\mathcal{NP}$ or not (i.e., the $\mathcal{NP}$ versus ${\rm co}\mathcal{NP}$ problem) and is a central topic in computational complexity theory, whose implications are profound for computer science, mathematics, and related fields.

To see the profound implications in various aspects of this open problem, let us dive into two cases: one is $\mathcal{NP}={\rm co}\mathcal{NP}$, and another is $\mathcal{NP}\ne{\rm co}\mathcal{NP}$. 

If $\mathcal{NP}={\rm co}\mathcal{NP}$, then for every problem where ``yes" answers are efficiently verifiable, ``no" answers would also be efficiently verifiable. This would be surprising because many problems, such as {\it Satisfiability} \cite{Coo71,Lev73} or {\it Graph non-isomorphism} \cite{GJ79}, seem to lack obvious short certificates for ``no" instances. It is also widely believed that $\mathcal{NP}={\rm co}\mathcal{NP}$ would imply significant consequences for the {\it Polynomial hierarchy} ($PH$) \cite{Sto77,A6}, a layered structure of complexity classes built on $\mathcal{NP}$ and ${\rm co}\mathcal{NP}$. Specifically, if $\mathcal{NP}={\rm co}\mathcal{NP}$, then $PH$ would collapse to $\mathcal{NP}$, i.e., $PH=\Sigma_1^P=\prod_1^P=\mathcal{NP}$; see e.g., Theorem 5.4 in \cite{AB09}. This contradicts the widely believed assumption that $PH$ is infinitely extended. Note that this open problem is not only one of the holy grail problems of theoretical computer science but also the core key to understanding the essence of computing.

Furthermore, many cryptographic systems rely on problems such as {\it Integer factorization} \cite{RSA78} and {\it Discrete logarithm} \cite{BG04} being hard to solve but easy to verify. If $\mathcal{NP}={\rm co}\mathcal{NP}$, then the ``no" instances might also become easy to verify, potentially weakening certain cryptographic assumptions. If $\mathcal{NP}={\rm co}\mathcal{NP}$, then optimization problems tied to $\mathcal{NP}$ (such as finding the shortest path in the {\it Traveling Salesman Problem} (TSP) \cite{Law85}) would have complements (e.g., proving no shorter path exists) that are equally tractable, which could revolutionize industries like logistics, where verifying optimality becomes as easy as finding solutions.

However, if $\mathcal{NP}\ne{\rm co}\mathcal{NP}$, it reinforces the belief that there is an asymmetry in computing: problems with efficiently verifiable ``yes" instances do not necessarily have efficiently verifiable ``no" instances. In particular, it would imply that $\mathcal{P}$ and $\mathcal{NP}$ differ. Furthermore, if $\mathcal{NP}\ne{\rm co}\mathcal{NP}$, the status quo in cryptography remains. The asymmetry ensures that problems can be hard to solve yet easy to verify in one direction, supporting the foundations of modern security protocols. Simultaneously, verifying the absence of a solution remains potentially intractable, meaning that we might efficiently find solutions but struggle to prove they are optimal, a common challenge in real-world applications.

Thus far, although the exact relationship between the complexity classes $\mathcal{NP}$ and ${\rm co}\mathcal{NP}$ is unknown and the complexity classes $\mathcal{NP}$ and ${\rm co}\mathcal{NP}$ are widely believed to be different, Figure \ref{fig1} below illustrates the possibility of the relationship between the complexity classes $\mathcal{NP}$ and ${\rm co}\mathcal{NP}$ that most complexity theorists believed.

\begin{figure}[ht]
\centering
\includegraphics[width=9cm]{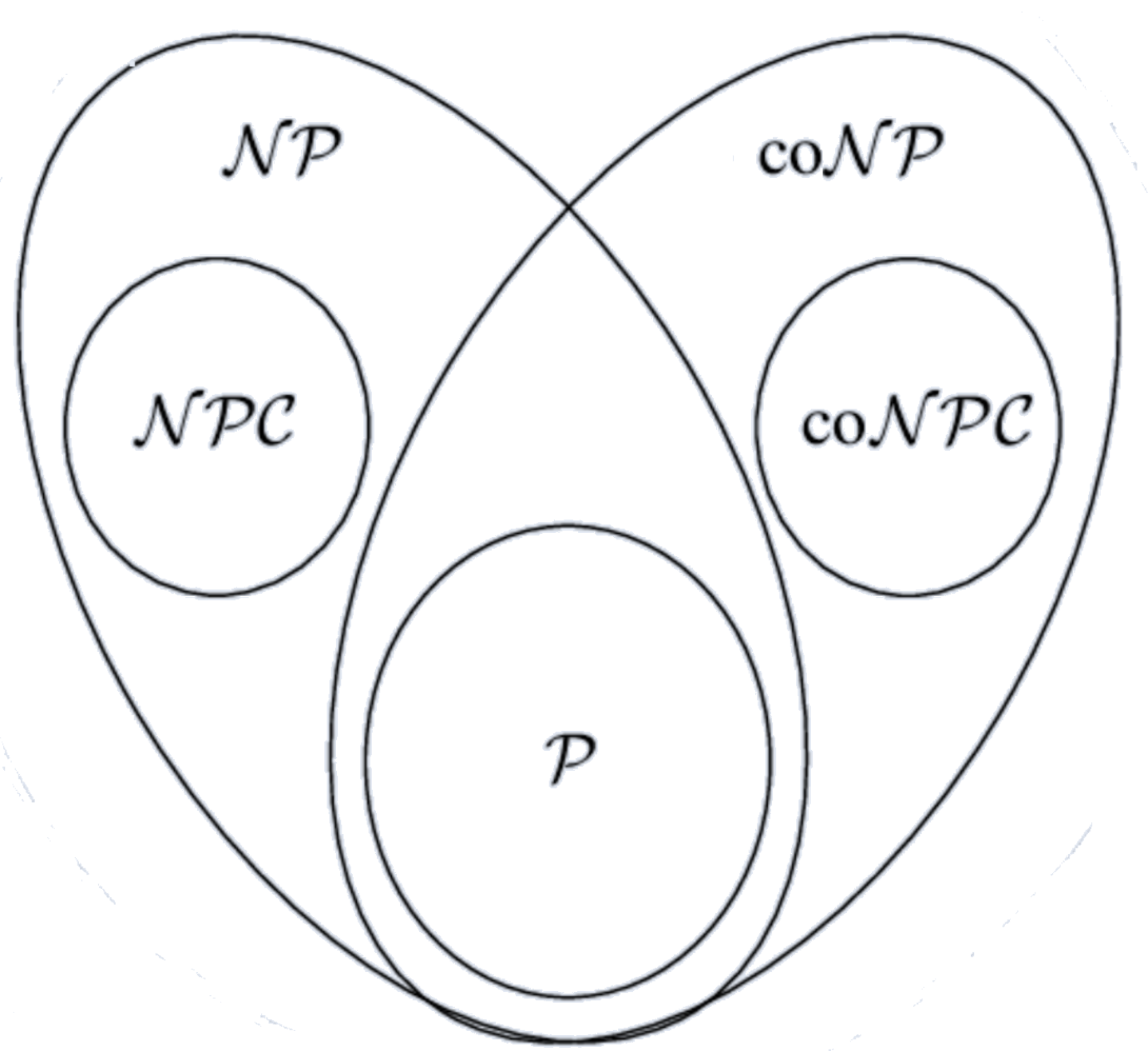}
\caption{\label{fig1}The most believed possibility between $\mathcal{NP}$ and ${\rm co}\mathcal{NP}$}
\end{figure}

Our motivation in this paper is that, in the author's recent work \cite{Lin21}, where we presented a proof that $\mathcal{P}\ne\mathcal{NP}$, we were completely at a loss and didn't know how to prove $\mathcal{NP}\ne{\rm co}\mathcal{NP}$, so we left this outstanding and long-standing conjecture of $\mathcal{NP}\ne{\rm co}\mathcal{NP}$ untouched. In the present paper, we try to resolve this important open conjecture. We would like to stress that, generally, the conjecture of $\mathcal{NP}\ne{\rm co}\mathcal{NP}$ is believed to be harder than the conjecture of $\mathcal{P}\ne\mathcal{NP}$, since the former easily implies the latter.

In addition, the $\mathcal{NP}\ne{\rm co}\mathcal{NP}$ conjecture is closely related to {\it proof complexity}. The study of proof systems was initially motivated by the $\mathcal{NP}$ versus ${\rm co}\mathcal{NP}$ problem. Cook and Reckhow \cite{CR79} define a proof system for a set $L$ as a polynomial-time computable function $f$ whose range is $L$. They show that $\mathcal{NP}={\rm co}\mathcal{NP}$ if and only if there exists a proof system with polynomial-size proofs for the set of propositional tautologies.

As our second goal, we will consider the proof systems that are familiar from textbook presentations of logic, such as the {\em Frege systems}. It can be said that the {\it Frege proof system} is a ``textbook-style" propositional proof system with Modus Ponens as its only rule of inference. Interestingly, there is also an important open problem asking to show a superpolynomial lower bound on the length of proofs for a Frege system or to prove that it is polynomially bounded. Namely, the following

\begin{center}

\fbox{\parbox{\textwidth} {Prove a superpolynomial lower bound on the length of proofs for a Frege system (or prove that it is polynomially bounded).}}

\end{center}

Indeed, the above problem was listed as the first open problem in \cite{Pud08}. However, despite decades of effort, very little progress has been made on this long-standing open problem. Note that, previously, although $\mathcal{NP}\ne{\rm co}\mathcal{NP}$ is considered to be very likely true, researchers are not able to prove that some very basic proof systems are not polynomially bounded; see e.g., \cite{Pud08}. Our second goal is to resolve this long-standing problem.

\vskip 0.3 cm
\subsection{Our Contributions}

In this paper, we explore and settle the aforementioned important open problems. Our first main goal in this paper is to prove the following central theorem:
\begin{theorem}
\label{theorem1}
There exists a language $L_s$ accepted by a nondeterministic Turing machine but by no ${\rm co}\mathcal{NP}$ machines, i.e., $L_s\notin{\rm co}\mathcal{NP}$. Furthermore, it can be proved that $L_s\in\mathcal{NP}$. That is, $$\mathcal{NP}\ne{\rm co}\mathcal{NP}. $$
\end{theorem}

The notion of ${\rm co}\mathcal{NP}$ machines appearing in Theorem \ref{theorem1} is defined in Section \ref{sec:preliminaries}.

As mentioned earlier, Theorem \ref{theorem1} has profound implications, one of which is the following important corollary.
\begin{corollary}
\label{corollary2}
$\mathcal{P}\ne\mathcal{NP}$.\footnote{Thus, by the end of this paper, it will be clear to the reader that resolving the $\mathcal{P}\overset{?}{=}\mathcal{NP}$ problem does not necessarily require diagonalization techniques.}
\end{corollary}

Thus, we once again arrive at the important result that $\mathcal{P}\ne\mathcal{NP}$ \cite{Lin21}.

Furthermore, Theorem \ref{theorem1} also implies that polynomial-time many-one reduction \cite{Kar72} is more precise than polynomial-time Turing reduction \cite{Coo71}. Under polynomial-time Turing reduction, every ${\rm co}\mathcal{NP}$ language is polynomial-time Turing reducible to an $\mathcal{NP}$ language, i.e., ${\rm co}\mathcal{NP}\subseteq\mathcal{NP}$. However, in fact, the complement of $L_s$ is in ${\rm co}\mathcal{NP}$ (i.e., $\overline{L_s}\in{\rm co}\mathcal{NP}$) but is not in $\mathcal{NP}$ (i.e., $\overline{L_s}\notin\mathcal{NP}$) by Theorem \ref{theorem1}.

Note that we prove the first half of Theorem \ref{theorem1} using ordinary simulation techniques together with the relationship between languages $L(M)$ accepted by a polynomial-time nondeterministic Turing machine $M$ and the complement $\overline{L}(M)$ of $L(M)$. Furthermore, by the result of \cite{BGS75}, there exists some oracle $A$ such that $\mathcal{NP}^A={\rm co}\mathcal{NP}^A$. Just as the result that $\mathcal{P}^A=\mathcal{NP}^A$ \cite{BGS75} led to a strong belief that problems with contradictory relativization are very hard to solve and are not amenable to current proof techniques --- i.e., that the solutions of such problems are beyond the current state of mathematics (see e.g., \cite{HCCRR93, Hop84}) --- the conclusion $\mathcal{NP}^A={\rm co}\mathcal{NP}^A$ likewise suggests that separating $\mathcal{NP}$ from ${\rm co}\mathcal{NP}$ is beyond the current state of mathematics. Therefore, the next result, stated in the following theorem, breaks the so-called Relativization Barrier:
\begin{theorem}
\label{theorem7}
Under some reasonable assumptions (see Section \ref{sec:beating_the_relativization_barrier} below), if $\mathcal{P}^A\ne\mathcal{NP}^A={\rm co}\mathcal{NP}^A$, then the set of all ${\rm co}\mathcal{NP}^A$ machines (defined in Section \ref{sec:beating_the_relativization_barrier}) is not enumerable. Consequently, ordinary simulation techniques generally will {\em not} apply to the relativized versions of the $\mathcal{NP}$ versus ${\rm co}\mathcal{NP}$ problem.
\end{theorem}

There is also another case, namely $\mathcal{P}^A=\mathcal{NP}^A={\rm co}\mathcal{NP}^A$. In this scenario, we remark that Theorem \ref{theorem6.1} in Section \ref{sec:beating_the_relativization_barrier} fully explains the reason why the simulation techniques cannot be applied to separate $\mathcal{NP}^A$ from ${\rm co}\mathcal{NP}^A$.

Moreover, the complexity class $\mathcal{BQP}$ (bounded-error quantum polynomial-time; see page 213 of \cite{AB09}) has a similar property to $\mathcal{P}={\rm co}\mathcal{P}$ (i.e., $\mathcal{BQP}=\rm{co}\mathcal{BQP}$). We therefore have the first important conclusion concerning the relationship between the complexity classes $\mathcal{BQP}$ and $\mathcal{NP}$:
\begin{corollary}
\label{corollary1point4}
$\mathcal{BQP}\ne\mathcal{NP}$.\footnote{The precise relationships between $\mathcal{BQP}$ and $\mathcal{NP}$ remain open; see e.g., Section \ref{sec:fundamental_open_problem}.}
\end{corollary}

Furthermore, let $\mathcal{NEXP}$ and ${\rm co}\mathcal{NEXP}$ be the complexity classes defined by $$\mathcal{NEXP}\overset{\rm def}{=}\bigcup_{k\in\mathbb{N}_1}{\rm NTIME}[2^{n^k}] $$ and $${\rm co}\mathcal{NEXP}\overset{\rm def}{=}\bigcup_{k\in\mathbb{N}_1}{\rm coNTIME}[2^{n^k}], $$ respectively (the complexity classes ${\rm NTIME}$ and ${\rm coNTIME}$ are defined in Section \ref{sec:preliminaries}). Then, in addition to the above corollary, Theorem \ref{theorem1} yields the following consequence, whose proof is similar to that of Theorem \ref{theorem1}:

\begin{corollary}
\label{corollary3}
 $\mathcal{NEXP}\neq {\rm co}\mathcal{NEXP}$.
\end{corollary}

Since $\mathcal{EXP}={\rm co}\mathcal{EXP}$ (where $\mathcal{EXP}$ is the class of languages accepted by deterministic Turing machines in time $2^{n^{O(1)}}$; see Section \ref{sec:preliminaries} for big $O$ notation), this, together with Corollary \ref{corollary3}, implies that $\mathcal{EXP}\ne\mathcal{NEXP}$. Moreover, the complexity class $\mathcal{BPP}$ (bounded-error probabilistic polynomial-time) aims to capture the set of decision problems efficiently solvable by polynomial-time probabilistic Turing machines (see e.g., \cite{AB09}). A central open problem in probabilistic complexity theory is the relation between $\mathcal{BPP}$ and $\mathcal{NEXP}$. Currently, researchers only know that $\mathcal{BPP}$ is sandwiched between $\mathcal{P}$ and $\mathcal{EXP}$ (i.e., $\mathcal{P}\subseteq\mathcal{BPP}\subseteq\mathcal{EXP}$), but they have been unable to show that $\mathcal{BPP}$ is a proper subset of $\mathcal{NEXP}$; see e.g., page 126 of \cite{AB09}. With the consequence of Corollary \ref{corollary3} (i.e., $\mathcal{EXP}\ne\mathcal{NEXP}$) in hand, it immediately follows that
\begin{corollary}
\label{corollary7}
$\mathcal{BPP}\ne\mathcal{NEXP}$.
\end{corollary}

It is interesting that the complexity class $\mathcal{NP}$ has a rich structure if $\mathcal{P}$ and $\mathcal{NP}$ differ. Specifically, in \cite{Lad75}, Ladner constructed a language that is $\mathcal{NP}$-intermediate under the assumption that $\mathcal{P}\neq\mathcal{NP}$. In fact, since by Theorem \ref{theorem1}, we know that $\mathcal{P}\ne{\rm co}\mathcal{NP}$, we also obtain a symmetric result showing that the complexity class ${\rm co}\mathcal{NP}$ contains intermediate languages. To see this, let $L\in\mathcal{NP}$ be the $\mathcal{NP}$-intermediate language constructed in \cite{Lad75}; then $\overline{L}=\{0,1\}^*\setminus L$ is a ${\rm co}\mathcal{NP}$-intermediate language in ${\rm co}\mathcal{NP}$. In other words, we will prove the following important result in detail:

\begin{theorem}
\label{theorem4}
There exist ${\rm co}\mathcal{NP}$-intermediate languages; that is, there exists a language $L\in{\rm co}\mathcal{NP}$ that is neither in $\mathcal{P}$ nor ${\rm co}\mathcal{NP}$-complete.
\end{theorem}

By the result of \cite{CR74}, there exists a super proof system if and only if $\mathcal{NP}$ is closed under complement. Then, by Theorem \ref{theorem1}, we immediately have the following:

\begin{corollary}
\label{corollary5}
There exists no super proof system. 
\end{corollary}

Moreover, we will settle the open problem listed first in \cite{Pud08}; see the introduction in Section \ref{sec:introduction}. Namely, we prove the following theorem concerning the aforementioned problem:
\begin{theorem}
\label{theorem6}
There exists no polynomial $p(n)$ such that for all $\psi\in TAUT$, there is a Frege proof of $\psi$ of length at most $p(|\psi|)$. In other words, no Frege proof systems of $\psi\in TAUT$ is polynomially bounded.
\end{theorem}

\vskip 0.3 cm
\subsection{Background and Prior Work}

In an epoch-making $36$-page paper \cite{Tur37} published in 1936, Turing laid the foundations of computer science. Turing's contributions were so influential in the development of theoretical computer science that he is widely regarded as its father \cite{A7}. However, although Turing's work initiated the study of theoretical computer science, he was not concerned with the efficiency of his machines--the central focus of {\em computational complexity theory}. In fact, Turing's concern \cite{Tur37} was whether his machines could simulate arbitrary algorithms given sufficient time (see e.g., \cite{Coo00}).

Why are some problems easy for computers to solve while others seem fundamentally intractable? This question lies at the heart of {\it computational complexity theory}. So what is computational complexity theory? {\it Computational complexity theory} is a central subfield of {\em theoretical computer science} and {\it mathematics} that primarily concerns the efficiency of Turing machines (i.e., algorithms) and the intrinsic complexity of computational tasks. In other words, it focuses on classifying computational problems according to their resource usage and relating these classes to each other (see e.g., \cite{A1}). Computational complexity theory specifically addresses fundamental questions such as what is feasible computation and what can and cannot be computed with a reasonable amount of computational resources such as time or space (memory). The field contains many remarkable open problems, most notably the famous $\mathcal{P}$ versus $\mathcal{NP}$ problem and the important open problem concerning the unknown relationship between the complexity classes $\mathcal{NP}$ and ${\rm co}\mathcal{NP}$. These problem continue to attract intense research interest. However, despite decades of effort, very little progress has been made. As Wigderson pointed out\cite{Wig07}, understanding the power and limits of efficient computation has driven the development of {\em computational complexity theory}. This discipline in general --- and the aforementioned famous open problems in particular --- has gained prominence within the mathematics community over the past decades.

As noted in Wikipedia \cite{A1}, in computational complexity theory a problem is regarded as inherently difficult if its solution requires significant resources (such as time and space) no matter which algorithm used. The theory formalizes this intuition by introducing mathematical models of computation and quantifying the resources need to solve problems. Other complexity measure are also studied, such as the amount of communication in communication complexity and the number of gates in circuit complexity (see e.g.,\cite{A1}). One of the central roles of computational complexity theory is to determine the practical limits on what computers (the computing models) can and cannot do. 

Historically, the fundamental measure of time led to the study of the extremely expressive complexity class $\mathcal{NP}$, one of the most important classical complexity classes (i.e., nondeterministic polynomial-time). Specifically, $\mathcal{NP}$ is the set of decision problems for which the problem instances, where the answer is ``yes," have proofs verifiable in polynomial time by some deterministic Turing machine. Equivalently, it is the class of problems solvable in polynomial time by a nondeterministic Turing machine (see e.g., \cite{A2}). The famous Cook-Levin theorem \cite{Coo71,Lev73} shows that this class has complete problems: it proves that {\em Satisfiability} (SAT) is $\mathcal{NP}$-complete, meaning {\em Satisfiability} is in $\mathcal{NP}$ and every other language in $\mathcal{NP}$ can be reduced to it in deterministic polynomial time. This famous and seminal result stimulated Karp's influential work \cite{Kar72}, which opened the door to the research of the renowned rich theory of $\mathcal{NP}$-completeness \cite{Kar72}.

On the other hand, the complexity class ${\rm co}\mathcal{NP}$ is the complement of $\mathcal{NP}$. Formally, for a complexity class $\mathcal{C}$, its complement, denoted ${\rm co}\mathcal{C}$, is defined as $${\rm co}\mathcal{C}=\{\overline{L}:L\in\mathcal{C}\}, $$where $L$ is a decision problem, and $\overline{L}$ is the complement of $L$ (i.e., $\overline{L}=\Sigma^*\setminus L$, assuming the language $L$ is over the alphabet $\Sigma$) \cite{Pap94}. Note that the complement of a decision problem $L$ is the decision problem whose answer is ``{\em yes}'' whenever the input is a ``{\em no}'' input of $L$, and vice versa. Equivalently, according to Wikipedia \cite{A3}, ${\rm co}\mathcal{NP}$ consists of those decision problems for which there exists a polynomial $p(n)$ and a polynomial-time bounded Turing machine $M$ such that, for every instance $x$, $x$ is a no-instance if and only if for some possible certificate $c$ of length bounded by $p(n)$, the Turing machine $M$ accepts the pair $(x,c)$. To the best of our knowledge, the complexity class ${\rm co}\mathcal{NP}$ was introduced for the first time by Meyer and Stockmeyer \cite{MS72} under the name ``$\prod_1^P$," and Stockmeyer wrote a full paper on the polynomial hierarchy that also uses the notation ${\rm co}\mathcal{NP}$ (see e.g., \cite{Sto77}).

The $\mathcal{NP}$ versus ${\rm co}\mathcal{NP}$ problem is deeply connected to the field of proof complexity. In 1979, Cook and Reckhow \cite{CR79} introduced a general definition of a propositional proof system and related it to mainstream complexity theory by pointing out that such a system exists in which all tautologies have polynomial-length proofs if and only if the two complexity classes $\mathcal{NP}$ and ${\rm co}\mathcal{NP}$ coincide; see e.g., \cite{CN10}. For further background on the importance of this research area and the motivation of the development of this rich theory, see \cite{Coo00} and \cite{Kra95}. Chapter $10$ of \cite{Pap94} also discusses the significance of the problem $$\mathcal{NP}\overset{?}{=}{\rm co}\mathcal{NP}. $$

In the next few paragraphs, we recall some of the interesting background, the current status, history, and main goals of the field of {\it proof complexity}. 

{\em Proof theory} is a major branch of {\em mathematical logic} and {\em theoretical computer science} in which proofs are treated as formal mathematical objects that can be analyzed using mathematical techniques (see e.g., \cite{A4}). Its major subareas include {\em structural proof theory}, {\em ordinal analysis}, {\em automated theorem proving}, and {\em proof complexity}. Proof theory also has important applications in computer science, linguistics, and philosophy. As Razborov noted \cite{Raz04}, proof and computations are among the most fundamental concepts in human intellectual activity; both have been central to the development of mathematics for thousands of years. The effort to study these concepts themselves in a rigorous, metamathematical way began in the $20$th century and led to the flourishing of mathematical logic and its derived disciplines. 

According to \cite{Rec76}, logicians have proposed a great number of systems for proving theorems. These systems provide certain rules for constructing proofs and for associating a theorem (formula) with each proof. More importantly, these rules are much simpler to understand than the theorems themselves. Thus, a proof offers a constructive way of establishing that a theorem is true. In addition, a proof system is {\em sound} if every theorem is true, and it is {\em complete} if every true statement (from a certain class) is a theorem (i.e., has a proof); see e.g., \cite{Rec76}. 

Within proof theory, proof complexity is the subfield aiming to understand and analyze the computational resources that are required to prove or refute statements (see e.g., \cite{A5}). Research in proof complexity is predominantly concerned with proving proof-length lower and upper bounds in various propositional proof systems. One of its major challenges is to show that the Frege system (standard propositional calculus) does not admit polynomial-size proofs of all tautologies (see e.g., \cite{A5}). Here, the size of a proof is simply the number of symbols it contains, and a proof is said to be of polynomial size if it is polynomial in the size of the tautology it proves. Contemporary proof complexity research draws ideas and methods from many areas in computational complexity, algorithms, and mathematics. Since many important algorithms and algorithmic techniques can be viewed as proof search procedures for certain proof systems, proving lower bounds on proof sizes in these systems implies run-time lower bounds on the corresponding algorithms; see e.g., \cite{A5}.

{\em Propositional proof complexity} has developed rapidly over the last two decades. As Razborov observed \cite{Raz15}, it plays a role in the theory of feasible proofs analogous to the role that Boolean circuits plays in the theory of efficient computations. In most cases, the central question of propositional proof complexity can be stated as follows: given a mathematical statement encoded as a propositional tautology $\psi$ and a class of admissible mathematical proofs formalized as a propositional proof system $P$, what is the minimal possible complexity of a $P$-proof of $\psi$ (see e.g., \cite{Raz15})? In other words, propositional proof complexity aims to understand and analyze the computational resources required to prove propositional tautologies, in the same way that circuit complexity studies the resources required to compute Boolean functions. From the perspective presented in \cite{Raz15}, proving lower bounds for strong proof systems such as Frege or Extended Frege --- modulo any hardness assumption from the purely computational world --- may be almost as interesting; see e.g., \cite{Raz15}. It is also worth noting that Razborov \cite{Raz15} observed that $\mathcal{NP}\ne{\rm co}\mathcal{NP}$ implies lower bounds for any propositional proof system. In addition, see, e.g., \cite{Raz03} for further interesting directions in propositional proof complexity.

Systematic study of proof complexity began with the seminal work of Cook and Reckhow \cite{CR79}, who provided the basic definition of a propositional proof system from a computational complexity perspective. As pointed out by \cite{FSTW21}, their work \cite{CR79} connects the goal of propositional proof complexity to fundamental hardness questions in computational complexity such as the $\mathcal{NP}$ versus ${\rm co}\mathcal{NP}$ problem: establishing super-polynomial lower bounds for every propositional proof system would separate $\mathcal{NP}$ from ${\rm co}\mathcal{NP}$.

The $\mathcal{NP}$ versus ${\rm co}\mathcal{NP}$ problem is central to proof complexity \cite{Kra19}. It formalizes the question of whether or not there are efficient ways to prove the negative cases in ours and in many other similar examples. The ultimate goal of {\em proof complexity} is to show that no universal propositional proof system allows efficient proofs of all tautologies, which is equivalent to showing that $\mathcal{NP}$ is not closed under the complement; see e.g., \cite{Kra19}.

In 1974, Cook and Reckhow \cite{CR74} showed that there exists a super proof system if and only if $\mathcal{NP}$ is closed under complement; that is, if and only if $\mathcal{NP}={\rm co}\mathcal{NP}$. This result gave rise to ``Cook's program" for proving $\mathcal{NP}\ne {\rm co}\mathcal{NP}$ (which can also be viewed as a program for proving $\mathcal{P}\ne\mathcal{NP}$): establish  superpolynomial lower bounds for proof lengths in stronger and stronger proposition proof systems until they hold for all abstract proof systems. Although attractive in principle, this approach has proven extremely difficult in practice \cite{Bus12}.
\vskip 0.3 cm
\subsection{Our Approach}

The aforementioned ``Cook's program" is one of the best-known approaches to resolving the long-standing open conjecture $\mathcal{NP}\ne{\rm co}\mathcal{NP}$.

In this paper, we take a different route and apply simulation techniques. To the best of our knowledge, the standard way to construct a language outside a given complexity class $\mathcal{C}$ is diagonalization (or lazy-diagonalization) (see e.g., \cite{Tur37, HS65, AHU74, AB09}), which remains the only known method for proving lower bounds against uniform models such as Turing machines. However, we found it difficult (or even impossible at all) to directly apply diagonalization or lazy-diagonalization to produce a language outside ${\rm co}\mathcal{NP}$. Instead, by exploiting the relationship between the language $L(M)$ accepted by a polynomial-time nondeterministic Turing machine $M$ and its complement $\overline{L}(M)$ accepted by the corresponding ${\rm co}\mathcal{NP}$ machine $M$ (defined below), we rely solely on simulation techniques.

We first show that all polynomial-time nondeterministic Turing machines are enumerable. We then construct a universal nondeterministic Turing machine that simulates all polynomial-time nondeterministic Turing machines, thereby creating a language outside ${\rm co}\mathcal{NP}$. This simulation-based approach is novel: it is the first time such a method has been used for this purpose, and it does not appear in prior literature or standard textbooks such as \cite{AHU74,AB09,Sip13,HU79,Pap94}.

We further prove that this language actually belongs to $\mathcal{NP}$ using the novel technique introduced in the author's recent work \cite{Lin21}, thereby establishing the desired separation $\mathcal{NP}\neq{\rm co}\mathcal{NP}$. Prior to this work, no method was known to resolve this fundamental and important open conjecture.

\vskip 0.3 cm
\subsection{Organization}

The remainder of this paper is organized as follows. For the reader's convenience, the next section reviews some basic notions closely related to our discussions, fixes the notation used throughout the paper, and presents several useful technical lemmas. In Section \ref{sec:all_polynomial_time_nondeterministic_turing_machines_are_enumerable}, we prove that all polynomial-time nondeterministic Turing machines are enumerable, and hence that all ${\rm co}\mathcal{NP}$-machines are enumerable as well. This establishes the existence of an enumeration of all ${\rm co}\mathcal{NP}$-machines. In Section \ref{sec:simulation_of_all_polynomial-time_nondeterministic_turing_machines}, we present the definition of our universal nondeterministic Turing machine. This machine accepts a language $L_s$ that is not in ${\rm co}\mathcal{NP}$, thereby establishing the desired lower bounds. Section \ref{sec:showing_l_s_in_np} is devoted to showing that the language $L_s$ is in fact in $\mathcal{NP}$, which provides the required upper bounds and completes the proof of Theorem \ref{theorem1}. In Section \ref{sec:beating_the_relativization_barrier}, we prove Theorem \ref{theorem7}, thereby breaking the so-called Relativization Barrier. In Section \ref{sec:structure_of_coNP}, we show the existence of ${\rm co}\mathcal{NP}$-intermediate languages. In Section \ref{sec:frege_systems}, we prove Theorem \ref{theorem6}, which states that no Frege proof system is polynomially bounded, answering a long-standing and important open question in proof complexity. Finally, we list some open problems for future research in the last section.

\vskip 0.3 cm
\section{Preliminaries}
\label{sec:preliminaries}
\vskip 0.3 cm

In this section, we introduce some notions and notation that will be used in what follows.

Let $\mathbb{N}_1=\{1,2,3,\cdots\}$ be the set of all positive integers ($+\infty\notin \mathbb{N}_1$). We also denote by $\mathbb{N}$ the set of natural numbers,\footnote{For the natural number system and $\infty$ (i.e., $+\infty$ here), we refer the reader to the interesting textbook \cite{Tao22}, in which Tao provides a full and philosophical treatment of them.} i.e., $$\mathbb{N}\overset{\rm def}{=}\mathbb{N}_1\cup\{0\}. $$

Let $\Sigma$ be an alphabet. For finite words $w,v\in\Sigma^*$, the concatenation of $w$ and $v$, denoted by $w\circ v$, is $wv$. For example, suppose that $\Sigma=\{0,1\}$, $w=100001$ and $v=0111110$; then 

$$\aligned
 w\circ v\overset{\rm def}{=}&wv\\
=&1000010111110
\endaligned$$

The length of a finite word $w$, denoted by $|w|$, is defined to be the number of symbols in it. It is clear that for finite words $w$ and $v$, $$|wv|=|w|+|v|.$$

The big $O$ notation indicates the order of growth of a quantity as a function of $n$, or the limiting behavior of a function. More precisely, the statement that $S(n)$ is big $O$ of $f(n)$, written $$S(n)=O(f(n)),$$ means that there exists a positive integer $N_0$ and a positive constant $M$ such that $$S(n)\leq M\times f(n)$$ for all $n>N_0$.

The little $o$ notation also indicates the order of growth of a quantity as a function of $n$, or the limiting behavior of a function, but with different meaning. Specifically, the statement that $T(n)$ is little-$o$ of $t(n)$,written $$T(n)=o(t(n)), $$ means that for any constant $c>0$ there exists a positive integer $N_0>0$ such that $$T(n)<c\times t(n)$$ for all $n>N_0$.

Throughout this paper, the computational modes used are {\em nondeterministic Turing machines} (or their variants, such as {\em nondeterministic Turing machines with oracle}). We follow the standard definition of a nondeterministic Turing machine given in the standard textbook \cite{AHU74}. We first introduce the precise definition of a nondeterministic Turing machine as follows:

\begin{definition}[$k$-tape nondeterministic Turing machine, \cite{AHU74}]
\label{definition2.1}
A $k$-tape nondeterministic Turing machine (shortly, NTM) $M$ is a seven-tuple $(Q,T,I,\delta,\mathbbm{b},q_0,q_f)$
where:
\begin{enumerate}
\item {$Q$ is the set of states.}
\item {$T$ is the set of tape symbols.}
\item {$I$ is the set of input symbols; $I\subseteq T$.}
\item {$\mathbbm{b}\in T-I$, is the blank.}
\item {$q_0$ is the initial state.}
\item {$q_f$ is the final (or accepting) state.}
\item {$\delta$ is the next-move function, or a mapping from $Q\times T^k$ to subsets of 
$$
Q\times(T\times\{L,R,S\})^k.
$$
Suppose
$$
  \delta(q,a_1,a_2,\cdots,a_k)=\{(q_1,(a^1_1,d^1_1),(a^1_2,d^1_2),\cdots,(a^1_k,d^1_k)),\cdots,(q_n,(a^n_1,d^n_1),(a^n_2,d^n_2),\cdots,(a^n_k,d^n_k))\}
$$
and the nondeterministic Turing machine is in state $q$ with the $i$th tape head scanning tape symbol $a_i$ for $1\leq i\leq k$. Then in one move the nondeterministic Turing machine enters state $q_j$, changes symbol $a_i$ to $a^j_i$, and moves the $i$th tape head in the direction $d^j_i$ for $1\leq i\leq k$ and $1\le j\le n$.}
\end{enumerate}
\end{definition}

Let $M$ be a nondeterministic Turing machine and let $w$ be an input. Then we write $M(w)$ to denote the computation of $M$ is on input $w$. 

A nondeterministic Turing machine $M$ works in time $t(n)$ (or is of time complexity $t(n)$) if, for any input $w\in I^*$ (where $I$ is the input alphabet of $M$), the computation $M(w)$ halts within $t(|w|)$ steps.

We now introduce the notion of a polynomial-time nondeterministic Turing machine.

\begin{definition}[cf. polynomial-time deterministic Turing machines in \cite{Coo00}]
\label{definition2.2}
Formally, a polynomial-time nondeterministic Turing machine is a nondeterministic Turing machine $M$ for which there exists a $k\in\mathbb{N}_1$ such that, for all input $w$ of length $|w|\in\mathbb{N}$, the computation $M(w)$ halts within $t(|w|)=|w|^k+k$ steps. The language accepted by such a machine $M$ is denoted by $L(M)$.
\end{definition}

By default, a word $w$ is accepted by a polynomial-time (say, $t(n)$ time-bounded) nondeterministic Turing machine $M$ if there exists at least one computation path $\gamma$ that leads the machine into an accepting state (i.e., stopping in the ``accepting" state) on input $w$ (in this case we say $M$ accepts $w$). Thus, if $M$ accepts the language $L(M)$, then $w\in L(M)$ if and only if $M$ accepts $w$; that is, there exists at least one accepting path of $M$ on input $w$. This is the ``there exists at least one accepting path" condition for nondeterministic Turing machines, on the basis of which the complexity class ${\rm NTIME}[t(n)]$ is defined:
 
\begin{definition}[``there exists at least one accepting path" condition]
\label{definition2.3}
The set of languages decided by nondeterministic Turing machines within time $t(n)$ under the ``there exists at least one accepting path" condition is denoted by ${\rm NTIME}[t(n)]$. Under this condition, a $t(n)$ time-bounded nondeterministic Turing machine accepts a string $w$ if and only if there exists at least one accepting path on input $w$. Thus, $$\mathcal{NP}=\bigcup_{k\in\mathbb{N}_1}{\rm NTIME}[n^k]. $$
\end{definition}

In addition to the ``there exists at least one accepting path" condition used to define the complexity class ${\rm NTIME}[t(n)]$, there is a dual condition: that there exists no computation path $\gamma$ leading the machine into an accepting state. We call this the ``there exists no accepting path" condition.

Let $M$ be a polynomial-time (say, $t(n)$ time-bounded) nondeterministic Turing machine over the input alphabet $\{0,1\}$. Under the ``there exists no accepting path" condition, a word $w$ is rejected by $M$ if there exists no computation path $\gamma$ that leads the machine into an accepting state on input $w$ (in this case we say $M$ rejects $w$). Clearly, the machine $M$ then defines the language $\overline{L}(M)$, which is the complement of $L(M)$ (i.e., $\overline{L}(M)=\{0,1\}^*\setminus L(M)$). Specifically, $w\in\overline{L}(M)$ if and only if $M$ rejects $w$.

Since ${\rm co}\mathcal{NP}=\{\overline{L}\,:\,L\in\mathcal{NP}\}$, a polynomial-time nondeterministic Turing machine $M$ that accepts the language $L$ under the ``there exists at least one accepting path" condition simultaneously defines the language $\overline{L}$ under the ``there exists no accepting path" condition. We therefore define a ${\rm co}\mathcal{NP}$ machine as follows:

\begin{definition}
\label{definition2.4}
A ${\rm co}\mathcal{NP}$ machine $M$ is a polynomial-time nondeterministic Turing machine operating under the ``there exists no accepting path" condition. The language accepted by the ${\rm co}\mathcal{NP}$ machine $M$ is defined to be $\overline{L}(M)$.\footnote {In short, there is an interesting relationship: the ${\rm co}\mathcal{NP}$ machine $M$ accepts the language $\overline{L}(M)$, while the same machine $M$ viewed as a polynomial-time nondeterministic Turing machine rejects $\overline{L}(M)$; conversely, the polynomial-time nondeterministic Turing machine $M$ accepts $L(M)$, while the ${\rm co}\mathcal{NP}$ machine $M$ rejects $L(M)$.} That is, for any input $w$, we have $w\in\overline{L}(M)$ if and only if the polynomial-time nondeterministic Turing machine $M$ rejects $w$.
\end{definition}

In particular, if a $t(n)$ time-bounded ${\rm co}\mathcal{NP}$ machine $M$ (over the input alphabet $\{0,1\}$) accepts the language $\overline{L}(M)$, then $w\notin\overline{L}(M)$ if and only if the corresponding polynomial-time (i.e., $t(n)$ time-bounded) nondeterministic Turing machine $M$ accepts $w$.

For a ${\rm co}\mathcal{NP}$ machine $M$ we define the characteristic function  
$$\chi_{\overline{L}(M)}:\{0,1\}^*\rightarrow\{0,1\} $$ by
$$
\chi_{\overline{L}(M)}(w)=\left\{
                                 \begin{array}{ll}
                                   1, & \hbox{$w\in\overline{L}(M)$;} \\
                                   0, & \hbox{$w\notin\overline{L}(M)$.}
                                 \end{array}
                               \right.
$$

Let $M$ be a polynomial-time nondeterministic Turing machine. Since $M$ may be viewed either as a polynomial-time nondeterministic Turing machine operating under the ``there exists at least one accepting path" condition or as a ${\rm co}\mathcal{NP}$ machine operating under the ``there exists no accepting path" condition, we carefully distinguish the two situations by means of notation. We write $$M(w)=1\quad\text{(respectively, $M(w)=0$)}$$ to mean that the polynomial-time nondeterministic Turing machine $M$ accepts (respectively, rejects) the input $w$; we write $$
\chi_{\overline{L}(M)}(w)=1\quad\text{ (respectively, $\chi_{\overline{L}(M)}(w)=0$)}
$$ to mean that ${\rm co}\mathcal{NP}$ machine $M$ accepts (respectively, rejects) the input $w$.

It is immediate that,  for all $w\in\{0,1\}^*$, $w\in\overline{L}(M)$ if and only if 
$$
M(w)=0,\quad\text{(i.e., the polynomial-time nondeterministic Turing machine $M$ rejects $w$),}
$$
and $w\notin\overline{L}(M)$ if and only if 
$$
M(w)=1,\quad\text{(i.e., the polynomial-time nondeterministic Turing machine $M$ accepts $w$).}
$$

As already noted, the ``there exists no accepting path" condition is dual to the ``there exists at least one accepting path" condition, on the basis of which we define the complexity class ${\rm coNTIME}[t(n)]$.

\begin{definition}[``there exists no accepting path" condition]
\label{definition2.5}
The set of languages decided by nondeterministic Turing machines within time $t(n)$ under the ``there exists no accepting path" condition is denoted by ${\rm coNTIME}[t(n)]$. Thus, $${\rm co}\mathcal{NP}=\bigcup_{k\in\mathbb{N}_1}{\rm coNTIME}[n^k]. $$
\end{definition}

Equivalently, the complexity class ${\rm co}\mathcal{NP}$ may be defined via polynomial-time deterministic Turing machines acting as verifiers with universal quantifiers over witnesses: 

\begin{definition}[\cite{AB09}, Definition 2.20]
\label{definition2.6}
For every $L\subseteq\{0,1\}^*$, we say that $L\in{\rm co}\mathcal{NP}$ if there exists a polynomial $p(n)$ and a deterministic polynomial-time Turing machine $M$ such that for every $x\in\{0,1\}^*$, $$x\in L\Leftrightarrow\forall u\in\{0,1\}^{p(|x|)},\quad M(x,u)=1. $$
\end{definition}

Concerning the relation between the time complexity of $k$-tape and single-tape nondeterministic Turing machines, we recall the following useful lemma, extracted from the standard textbook \cite{AHU74} (see Lemma 10.1 in \cite{AHU74}), which will play an important role later:

\begin{lemma}[Lemma 10.1 in \cite{AHU74}]
\label{lemma2.1}
If $L$ is accepted by a $k$-tape nondeterministic $T(n)$ time-bounded Turing machine, then $L$ is accepted by a single-tape nondeterministic $O(T^2(n))$ time-bounded Turing machine.\Q.E.D
\end{lemma}

Since ${\rm co}\mathcal{NP}$ machines are simply polynomial-time nondeterministic Turing machines operating under the ``there exists no accepting path" condition, an analogous statement holds for them:
\begin{lemma}
\label{lemma2.2}
 If $L$ is accepted by a $k$-tape $T(n)$ time-bounded ${\rm co}\mathcal{NP}$ machine, then $L$ is accepted by a single-tape $O(T^2(n))$ time-bounded ${\rm co}\mathcal{NP}$ machine. 
\end{lemma}
\begin{proof}
The proof is similar to that of Lemma \ref{lemma2.1} (see \cite{AHU74}); the details are therefore omitted.
\end{proof}

\vskip 0.2cm
The following result on efficient simulation by a universal nondeterministic Turing machine will be useful in proving our main results in Section \ref{sec:simulation_of_all_polynomial-time_nondeterministic_turing_machines}.

\begin{lemma}[\cite{AB09}]
\label{lemma2.3}
There exists a Turing machine $U$ such that for every $x,\alpha\in\{0,1\}^*$, $U(x,\alpha)=M_{\alpha}(x)$, where $M_{\alpha}$ denotes the Turing machine represented by $\alpha$. Moreover, if $M_{\alpha}$ halts on input $x$ within $T(|x|)$ steps, then $U(x,\alpha)$ halts within $cT(|x|)\log T(|x|)$ steps,\footnote{In this paper, $\log n$ stands for $\log_2 n$. Moreover, $T(n)\log T(n)$ can be relaxed to $T(n)^2$, and thus the counters (which are to count up to $n^{k+1}$) in Section \ref{sec:simulation_of_all_polynomial-time_nondeterministic_turing_machines} and Section \ref{sec:beating_the_relativization_barrier} can be set to $n^{2k+1}$ since $T(n)=n^k+k$ is a polynomial (the subsequent proof remains almost unchanged).} where $c$ is a constant independent of $|x|$ and depending only on $M_{\alpha}$'s alphabet size, number of tapes, and number of states.\Q.E.D
\end{lemma}

Further background material and notions will be introduced as need in the course of proving the main results stated in Section \ref{sec:introduction}.

\vskip 0.3 cm
\section{All Polynomial-time Nondeterministic Turing Machines Are Enumerable}
\label{sec:all_polynomial_time_nondeterministic_turing_machines_are_enumerable}
\vskip 0.3 cm

In this section, our main goal is to establish an important theorem asserting that all polynomial-time nondeterministic Turing machines and all ${\rm co}\mathcal{NP}$ machines are enumerable. This result is a prerequisite for applying the simulation technique.

Following Definition \ref{definition2.2}, a polynomial-time nondeterministic Turing machine --- say, an $n^k+k$ time-bounded nondeterministic Turing machine --- can be represented by the tuple $(M,k)$, where $M$ is the $n^k+k$ time-bounded nondeterministic Turing machine itself and $k$ is the degree of polynomial $|x|^k+k$ such that $M(x)$ halts within $|x|^k+k$ steps for every input $x$. We call such a positive integer $k$ the order of $(M,k)$.

By Lemma \ref{lemma2.1} (respectively, Lemma \ref{lemma2.2}), it suffices to consider single-tape nondeterministic Turing machines (respectively, single-tape ${\rm co}\mathcal{NP}$ machines). Consequently, throughout the remainder of this paper, the term ``polynomial-time nondeterministic Turing machine" will mean a single-tape polynomial-time nondeterministic Turing machine, and the term ``${\rm co}\mathcal{NP}$ machine" will mean a single-tape ${\rm co}\mathcal{NP}$ machine.

To obtain our main result we must {\em enumerate} all ${\rm co}\mathcal{NP}$ machines, so that we may refer to the $j$-th ${\rm co}\mathcal{NP}$ machine in the enumeration. Before doing so, however, we first {\it enumerate} all polynomial-time nondeterministic Turing machines, or equivalently, prove that the set of all polynomial-time nondeterministic Turing machines is enumerable.

In what follows we employ the encoding method presented in \cite{AHU74}, p. 407, to represent a single-tape nondeterministic Turing machine by an integer.

Without loss of generality we adopt the following assumptions concerning the representation of a single-tape nondeterministic Turing machine with input alphabet $\{0,1\}$ because that will be all we need:

\begin{enumerate}
\item {The states are named 
$$
q_1,q_2,\cdots,q_s
$$
for some $s$, with $q_1$ the initial state and $q_s$ the accepting state.}
\item {The input alphabet is $\{0,1\}$.}
\item {The tape alphabet is 
$$
\{X_1,X_2,\cdots,X_t\}
$$
for some $t$, where $X_1=\mathbbm{b}$, $X_2=0$, and $X_3=1$.}
\item {The next-move function $\delta$ is a list of quintuples of the form,
$$
\{(q_i,X_j,q_k,X_l,D_m),\cdots,(q_i,X_j,q_f,X_p,D_z)\}
$$
meaning that 
$$
\delta(q_i,X_j)=\{(q_k,X_l,D_m),\cdots,(q_f,X_p,D_z)\},
$$
and $D_m$ is the direction, $L$, $R$, or $S$, if $m=1,2$, or $3$, respectively. We assume this quintuple is encoded by the string 
$$
10^i10^j10^k10^l10^m1\cdots 10^i10^j10^f10^p10^z1.
$$
}
\item {The nondeterministic Turing machine itself is encoded by concatenating in any order the codes for each of the quintuples in its next-move function. Additional $1$'s may be prefixed to the string if desired. The result will be some string of $0$'s and $1$'s, beginning with $1$, which we can interpret as an integer.}
\end{enumerate}

We encode the order $k$ of the tuple $(M,k)$ by the string $$10^k1.$$ The complete encoding of the tuple $(M,k)$ is then the concatenation of the binary string representing $M$ with the string $10^k1$. The resulting binary string may again be interpreted as an integer. 
    
Under this encoding, any integer that cannot be decoded is assumed to represent the trivial Turing machine with an empty next-move function. For convenience, we write $\{(M,k)\}$ for the set of all polynomial-time nondeterministic Turing machines. The encoding just described yields a $(1,1)$ correspondence between the set $\{(M,k)\}$ of all polynomial-time nondeterministic Turing machines and $\mathbb{N}_1$ (once undecodable integers are identified with the trivial machine). In other words, the set $\{(M,k)\}$ of all polynomial-time nondeterministic Turing machines is enumerable. We denote the resulting enumeration by $e$.

Every single-tape polynomial-time nondeterministic Turing machine appears infinitely often in the enumeration $e$, since given a polynomial-time nondeterministic Turing machine $(M,k)$, we may prefix $1$'s at will to find larger and larger integers representing the same set of $(M,k)$. We denote the polynomial-time nondeterministic Turing machine corresponding to the integer $j$ by $\widehat{M}_j$, where $j$ is the integer value of the binary string encodes the tuple $(M,k)$.

We have therefore established the following important result:
\begin{theoremsection}
\label{theorem3.1}
There exists a one-to-one correspondence between $\mathbb{N}_1$ and the set $\{(M,k)\}$ of all polynomial-time nondeterministic Turing machines. Consequently, the set $\{(M,k)\}$ of all polynomial-time nondeterministic Turing machines is enumerable (i.e., every polynomial-time nondeterministic Turing machine appears in the list $e$), and all languages in $\mathcal{NP}$ are enumerable (with languages appearing multiple times).\Q.E.D
\end{theoremsection}

\vskip 0.3 cm
\begin{remark}
\label{remark2}
\begin{figure}[ht]
\centering
\includegraphics[width=11cm]{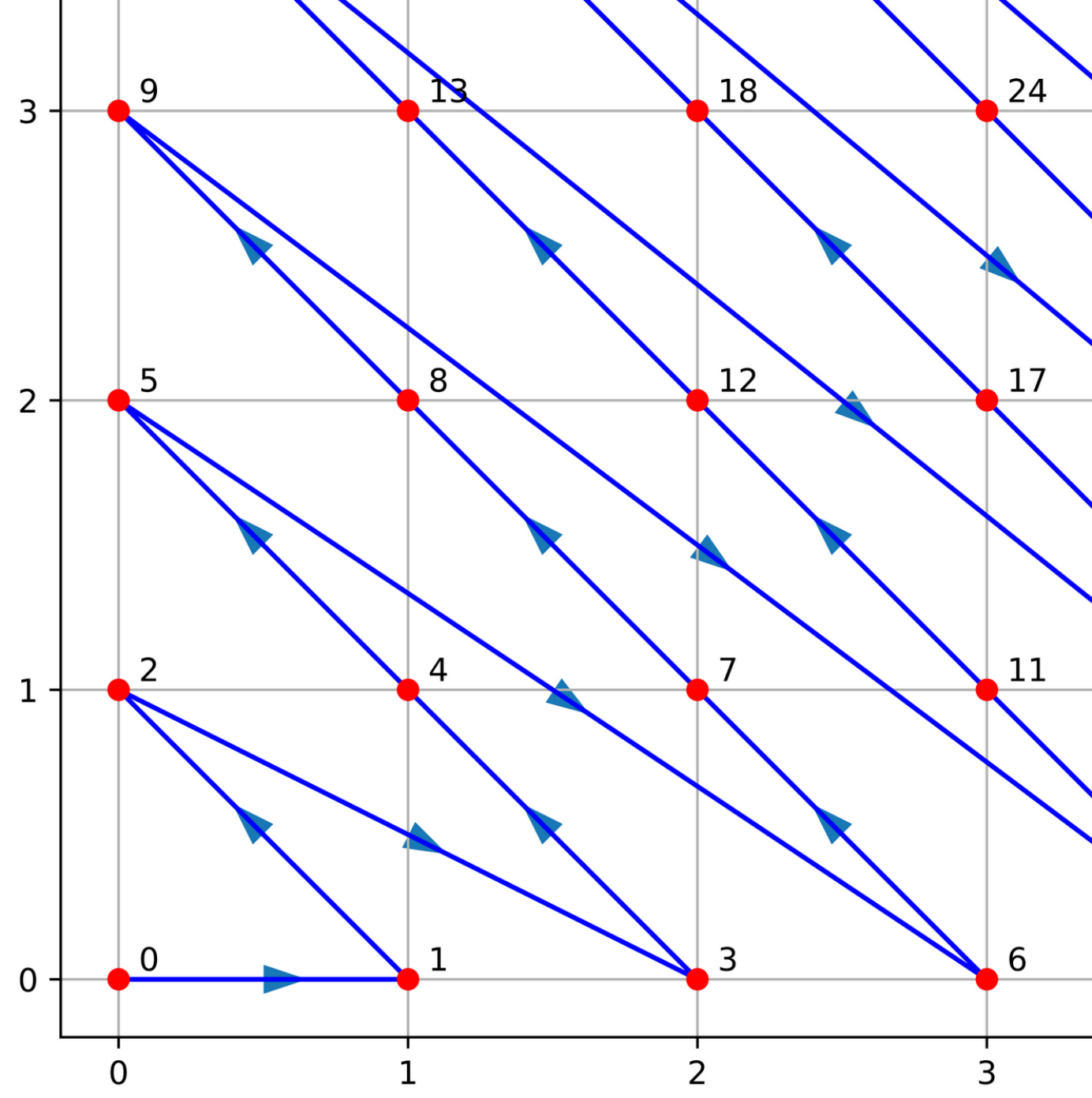}
\caption{\label{2}Cantor pairing function}
\end{figure}
There is an alternative way to {\em enumerate} all polynomial-time nondeterministic Turing machines that does not encode the polynomial order directly into the machine representation. It relies on  the {\em Cantor pairing function} (see Fig. \ref{2} above, taken from \cite{A8}): $$\pi:\mathbb{N}\times\mathbb{N}\rightarrow\mathbb{N}, $$ defined by $$\pi(k_1,k_2):=\frac{1}{2}(k_1+k_2)(k_1+k_2+1)+k_2, $$ where $k_1,k_2\in\mathbb{N}$. Since the Cantor pairing function is invertible (see \cite{A8}), it is a bijection between $\mathbb{N}\times\mathbb{N}$ and $\mathbb{N}$. We may further define a bijection $$\tau:\mathbb{N}\rightarrow\mathbb{N}_1$$ by $\tau(k)=k+1$, which establishes a one-to-one correspondence between $\mathbb{N}$ and $\mathbb{N}_1$. As we have already shown that every polynomial-time nondeterministic Turing machine can be identified with an integer, we may form the pair $(M,k)$ consisting of such a machine and its polynomial order, and then apply the Cantor pairing function to map this pair to an element of $\mathbb{N}_1$. The reader can easily check that the resulting map yields an enumeration of the set of all polynomial-time nondeterministic Turing machines.
\end{remark}

By Definition \ref{definition2.4}, the same enumeration $e$ also enumerates all ${\rm co}\mathcal{NP}$ machines. Indeed, it is enough to interpret every non-trivial machine appearing in $e$ under the ``there exists no accepting path" condition. Each polynomial-time nondeterministic Turing machine then corresponds uniquely to a ${\rm co}\mathcal{NP}$ machine, and vice versa. 

We therefore obtain the following theorem.

\begin{theoremsection}\footnote{In view of Definition \ref{definition2.4}, ${\rm co}\mathcal{NP}$ machines and $\mathcal{NP}$ machines share the same external representation (only the internal acceptance condition differs). Consequently the present theorem follows immediately from Theorem \ref{theorem3.1}.}
\label{theorem3.2}
 The set of all ${\rm co}\mathcal{NP}$ machines is enumerable (i.e., every ${\rm co}\mathcal{NP}$ machine appears in the list $e$). Equivalently, all languages in ${\rm co}\mathcal{NP}$ are enumerable (with languages appearing multiple times).\Q.E.D
\end{theoremsection}

In the sequel we shall frequently write $\langle M\rangle$ for a shortest binary string that encodes the $n^k+k$ time-bounded nondeterministic Turing machine $(M,k)$. Concretely, $\langle M\rangle$ is obtained by concatenating the codes of the quintuples of the next-move function (in any order) and appending the order string $10^k1$, without any additional leading $1$s. While $\langle M\rangle$ need not be unique --- different orders of the quintuples yield different strings --- all such shortest encodings have the same length. The same notation $\langle M\rangle$ will also be used for the corresponding shortest encoding of an $n^k+k$ time-bounded ${\rm co}\mathcal{NP}$ machine.

\vskip 0.3 cm
\section{Simulation of All Polynomial-time Nondeterministic Turing Machines}
\label{sec:simulation_of_all_polynomial-time_nondeterministic_turing_machines}
\vskip 0.3 cm

In the field of {\em computational complexity}, the most basic approach to showing hierarchy theorems uses simulation and diagonalization, because these techniques \cite{Tur37, HS65} are a standard method to prove lower bounds on uniform computing models (i.e., the Turing machine model); see for example \cite{AHU74}. These techniques work well for deterministic time and space measures; see e.g., \cite{For00, FS07}. However, they do not work for nondeterministic time, which is not known to be closed under complement at present --- one of the main issues discussed in this paper; hence it is unclear how to define a nondeterministic machine that ``does the opposite;" see e.g., \cite{FS07}.

For nondeterministic time, we can still perform a simulation but can no longer negate the answer directly. In this case, we apply the {\em lazy diagonalization} to prove the nondeterministic time hierarchy theorem; see e.g., \cite{AB09,SFM78,Zak83}. It is worth noting that Fortnow \cite{For11} developed a much more elegant and simple style of diagonalization to establish the nondeterministic time hierarchy; see also \cite{FS17}. Generally, {\em lazy diagonalization} \cite{Zak83, AB09, SFM78} is a clever application of the standard diagonalization technique \cite{Tur37, HS65, AHU74}. The basic strategy of {\em lazy diagonalization} is as follows: for any $n\le j<m=n^{t(n)}$, the machine $M$ on input $1^j$ simulates $M_i$ on input $1^{j+1}$, accepting if $M_i$ accepts and rejecting if $M_i$ rejects. When $j=m$, $M$ on input $1^j$ simulates $M_i$ on input $x$ deterministically, accepting if $M_i$ rejects and rejecting if $M_i$ accepts. Since $m=n^{t(n)}$, $M$ has enough time to perform the trivial exponential-time deterministic simulation of $M_i$ on input $1^n$; what we are doing is deferring the diagonalization step by linking $M$ and $M_i$ together on larger and larger inputs until $M$ has an input large enough that it can actually ``do the opposite" deterministically; see e.g., \cite{FS07} for details.

However, although the set of all ${\rm co}\mathcal{NP}$ machines has been proved to be enumerable and the enumeration is $e$ given in Section \ref{sec:all_polynomial_time_nondeterministic_turing_machines_are_enumerable}, it is challenging (and perhaps even impossible) to apply the standard diagonalization technique or the lazy-diagonalization technique from \cite{Zak83, AB09, FS07,SFM78} to create a language that is not in ${\rm co}\mathcal{NP}$. Instead, our main idea for this section is to apply simulation techniques rather than diagonalization techniques, together with the particular relations between the language $L(M)$ accepted by a polynomial-time nondeterministic Turing machine $M$ and the language $\overline{L}(M)$ accepted by a ${\rm co}\mathcal{NP}$ machine $M$ to show that there exists a language $L_s$ accepted by a universal nondeterministic Turing machine with the property that $L_s\notin{\rm co}\mathcal{NP}$.

We can now prove the following important theorem:
\begin{theoremsection}
\label{theorem4.1}
There exists a language $L_s$ accepted by a universal nondeterministic Turing machine $U$ but not by any ${\rm co}\mathcal{NP}$ machines. Namely, there is a language $L_s$ such that $$L_s\notin{\rm co}\mathcal{NP}. $$
\end{theoremsection}
\begin{proof}
Let $U$ be a four-tape nondeterministic Turing machine that operates as follows on an input string $x$ of length $n$.
\begin{enumerate}
\item{
$U$ decodes the tuple encoded by $x$. If $x$ is not the encoding of some ${\rm co}\mathcal{NP}$ machine $M_l$ for some $l$ (by Definition \ref{definition2.4}, equivalently, if $x$ cannot be decoded into some polynomial-time nondeterministic Turing machine, i.e., if $x$ represents the trivial Turing machine with an empty next-move function), then GOTO $5$; else determine $t$, the number of tape symbols used by $M_l$; $s$, its number of states; and $k$, its order. The third tape of $U$ can be used as ``scratch" memory to calculate $t$.
}
\item{
$U$ then lays off on its second tape $n$ blocks of $$\lceil\log t\rceil$$ cells each, the blocks being separated by a single cell holding a marker $\#$, i.e., there are $$(1+\lceil\log t\rceil)n$$ cells in all. Each tape symbol occurring in a cell of $M_l$'s tape will be encoded as a binary number in the corresponding block of the second tape of $U$. Initially, $U$ places $M_l$'s input, in binary-coded form, in the blocks of tape $2$, filling the unused blocks with the code for the blank.
}
\item{
On tape $3$, $U$ sets up a block of $$\lceil(k+1)\log n\rceil$$ cells, initialized to all $0$'s. Tape $3$ is used as a counter to count up to $$n^{k+1}. $$
}
\item{
Using nondeterminism, $U$ simulates $M_l$, using tape $1$ (its input tape) to determine the moves of $M_l$ and tape $2$ to simulate the tape of $M_l$. The moves of $U$ are counted in binary in the block of tape $3$, and tape $4$ is used to hold the states of $M_l$. If the counter on tape $3$ overflows, $U$ halts without accepting. The specified simulation is the following:

 ``On input $x=1^{n-m}\langle M_l\rangle$ where $m=|\langle M_l\rangle|$, $U$ on input $1^{n-m}\langle M_l\rangle$ simulates $M_l$ on input $1^{n-m}\langle M_l\rangle$ using nondeterminism in $$ n^{k+1} $$ time and outputs its answer, i.e., accepting if $M_l$ accepts and rejecting if $M_l$ rejects." 

\vskip 0.2cm
It should be pointed out that, although we say that the input $x$ is decoded into a ${\rm co}\mathcal{NP}$ machine $M_l$, $U$ in fact simulates the $n^k+k$ time-bounded nondeterministic Turing machine $M_l$. That is to say, $U$ accepts the input $x=1^{n-m}\langle M_l\rangle$ if there exists at least one accepting path of $M_l$ on input $x$, and $U$ rejects the input $x=1^{n-m}\langle M_l\rangle$ if there exists no accepting path of $M_l$ on input $x$.
}
\item{
Since $x$ is an encoding of the trivial Turing machine with an empty next-move function, $U$ sets up a block of 

$$\lceil 2\times\log n\rceil$$ cells on tape $3$, initialized to all $0$'s. Tape $3$ is used as a counter to count up to $$n^2.$$

Using its nondeterministic choices, $U$ moves according to the path given by $x$. The moves of $U$ are counted in binary in the block of tape $3$. If the counter on tape $3$ overflows, then $U$ halts. $U$ accepts $x$ if and only if there is a computation path from the start state of $U$ leading to the accept state and the total number of moves does not exceed $$n^2$$ steps, so the computation lies within $O(n)$ steps. Note that the factor $2$ in $$\lceil 2\times\log n\rceil$$ is fixed; it is a default setting.
}
\end{enumerate}

The nondeterministic Turing machine $U$ constructed above is of time complexity, say $S$, which is currently unknown. By Lemma \ref{lemma2.1}, $U$ is equivalent to a single-tape nondeterministic $O(S^2)$ time-bounded Turing machine, and it, of course, accepts some language $L_s$.

We claim that $L_s\notin{\rm co}\mathcal{NP}$. Indeed, suppose for the sake of contradiction that $L_s$ is decided by some $n^k+k$ time-bounded ${\rm co}\mathcal{NP}$ machine $M_i$. Then by Lemma \ref{lemma2.2}, we may assume that $M_i$ is a single-tape ${\rm co}\mathcal{NP}$ machine with time bound $T(n)=n^k+k$. By Definition \ref{definition2.4}, we may regard it as an $n^k+k$ time-bounded nondeterministic Turing machine. We further let $U$, on input $1^{n-m}\langle M_i\rangle$, simulate the $n^k+k$ time-bounded nondeterministic Turing machine $M_i$ on input $1^{n-m}\langle M_i\rangle$. Let $M_i$ have $s$ states and $t$ tape symbols, and let $m$ be the shortest length of $M_i$, i.e.,
$$
m=\min\{|\langle M_i\rangle|\,:\,\text{different binary strings $\langle M_i\rangle$ represent the same $M_i$}\}.
$$

Since $M_i$ is represented by infinitely many strings and appears infinitely often in the enumeration $e$, the machine $U$ sets the counter on tape $3$ to count up to $n^{k+1}$. By Lemma \ref{lemma2.3}, in order to simulate $$T(n)$$ steps of $M_i$ (note that the input to $M_i$ is $1^{n-m}\langle M_i\rangle$, where $m=|\langle M_i\rangle|$, the length of $1^{n-m}\langle M_i\rangle$ is $n$), $U$ needs to perform $$T(n)\log T(n)$$ steps. Moreover,

$$\aligned
\lim_{n\rightarrow\infty}&\frac{T(n)\log T(n)}{n^{k+1}}\\
=&\lim_{n\rightarrow\infty}\frac{(n^k+k)\log(n^k+k)}{n^{k+1}}\\
=&\lim_{n\rightarrow\infty}\left(\frac{n^k\log(n^k+k)}{n^{k+1}}+\frac{k\log(n^k+k)}{n^{k+1}}\right)\\
=&0\\
<&1.
\endaligned$$

Hence there exists an $N_0>0$ such that for any $N>N_0$, $$T(N)\log T(N)<N^{k+1}. $$This implies that for a sufficiently long $w$, say $|w|\ge N_0$, and $M_w$ denoted by such $w$ is $M_i$, we have $$T(|w|)\log T(|w|)<|w|^{k+1}. $$

Thus, on input $w$, $U$ has sufficient time to simulate $M_w$. Consequently,

$$1^{|w|-m}\langle M_i\rangle \in L_s\quad\text{if and and only if}\quad 1^{|w|-m}\langle M_i\rangle\in L(M_w)= L(M_i).\eqno(4.1) $$

By the relation between the language $L(M_i)$ accepted by polynomial-time nondeterministic Turing machine $M_i$ and the language $\overline{L}(M_i)$ accepted by the ${\rm co}\mathcal{NP}$ machine $M_i$, we also have $$1^ {|w|-m}\langle M_i\rangle\in L(M_i)\quad\text{if and only if}\quad 1^{|w|-m}\langle M_i\rangle\notin\overline{L}(M_i).\eqno(4.2) $$

From (4.1) and (4.2) we conclude that if $U$ accepts the input $1^{|w|-m}\langle M_i\rangle$, i.e., $$1^{|w|-m}\langle M_i\rangle\in L_s, $$ then  $$1^{|w|-m}\langle M_i\rangle\notin\overline{L}(M_i), $$ or equivalently, $$\chi_{\overline{L}(M_i)}(1^{|w|-m}\langle M_i\rangle)=0. $$ Likewise, if $U$ rejects the input $1^{|w|-m}\langle M_i\rangle$, i.e., $$1^{|w|-m}\langle M_i\rangle\notin L_s, $$ then $$1^{|w|-m}\langle M_i\rangle\in\overline{L}(M_i), $$ or equivalently, $$\chi_{\overline{L}(M_i)}(1^{|w|-m}\langle M_i\rangle)=1.$$ This contradicts our assumption that $L_s$ is decided by the $n^k+k$ time-bounded ${\rm co}\mathcal{NP}$ machine $M_i$. 

The above arguments show that no ${\rm co}\mathcal{NP}$ machine $M_i$ in the enumeration $e$ accepts the language $L_s$. Since all ${\rm co}\mathcal{NP}$ machines are in the list $e$, it follows that $$L_s\notin{\rm co}\mathcal{NP}. $$ This completes the proof.
\end{proof}

\vskip 0.3 cm
\begin{remark}
\label{remark4.1}
One might wish to apply the same simulation techniques to the question of $\mathcal{P}$ versus ${\rm co}\mathcal{P}$ and wonder why the simulation techniques are unable to produce a language not in ${\rm co}\mathcal{P}$. We remark that the simulation technique cannot create a language outside ${\rm co}\mathcal{P}$ via a universal deterministic Turing machine $U'$. The reason is that a polynomial-time deterministic Turing machine (a $\mathcal{P}$ machine) that operates under the condition ``there exists no accepting path" is still a polynomial-time deterministic Turing machine. In other words, the conditions ``there exists at least one accepting path" and ``there exists no accepting path" are symmetric for deterministic Turing machines. More precisely, for an arbitrary deterministic Turing machine $M$, if $M$ accepts the input $w$, then there is exactly one accepting path, and if $M$ rejects $w$ then there is exactly one rejecting path. This situation differs from that of nondeterministic Turing machines. Consequently, if we apply the same simulation technique in an attempt to create a language not in ${\rm co}\mathcal{P}$, the resulting universal deterministic Turing machine $U'$ (which runs in time $O(n^k)$ for any $k\in\mathbb{N}_1$) accepts a language $L(U')\in\mathcal{P}$ (the proof is analogous to the argument given in Section \ref{sec:showing_l_s_in_np} below). By the symmetry between the conditions ``there exists at least one accepting path" and ``there exists no accepting path" for deterministic Turing machines, we clearly immediately have that $L(U')\in{\rm co}\mathcal{P}$. Therefore we cannot claim that $L(U')\notin{\rm co}\mathcal{P}$.
\end{remark}

\vskip 0.3 cm
\section{Proving That $L_s\in\mathcal{NP}$}
\label{sec:showing_l_s_in_np}
\vskip 0.3 cm

The previous section aims to establish a lower bounds result, i.e., to show that there exists a language $L_s$ not in ${\rm co}\mathcal{NP}$. We now turn our attention to the corresponding upper-bound result, i.e., to show that the language $L_s$ accepted by the universal nondeterministic Turing machine $U$ in fact belongs to $\mathcal{NP}$.

As a matter of fact, the technique used to establish the desired upper-bound theorem is essentially the same as the one developed in the author's recent work \cite{Lin21}, because the underlying issue to establish desired upper bounds is essentially the same as that treated in \cite{Lin21}. But establishing such an upper bound for the first time was nevertheless very difficult; for an explanation of the difficulties involved, see e.g., \cite{Lin21}.

For our purposes we first show that the universal nondeterministic Turing machine $U$ runs within time $O(n^k)$ for any $k\in\mathbb{N}_1$:

\begin{theoremsection}
\label{theorem5.1}
The universal nondeterministic Turing machine $U$ constructed in the proof of Theorem \ref{theorem4.1} runs within time $O(n^k)$ for any $k\in\mathbb{N}_1$.
\end{theoremsection}

\begin{proof}
The simplest way to show the theorem is to prove that for any input $w$ to $U$, there exists a corresponding positive integer $i_w\in\mathbb{N}_1$ such that $U$ runs at most $$|w|^{i_w+1}$$ steps, which can be established as follows. 

On the one hand, if the input $x$ encodes an $n^k+k$ time-bounded nondeterministic Turing machine, then by construction $U$ turns itself off mandatorily within $$|x|^{k+1}$$ steps; thus the corresponding integer $i_x$ is $k$ (i.e., $i_x=k$). This holds for all polynomial-time nondeterministic Turing machines that appears as input, where $k$ is the degree of that corresponding polynomial of the polynomial-time nondeterministic Turing machine. 

On the other hand, if the input $x$ does not encode any polynomial-time nondeterministic Turing machine, then by construction $U$ is forced to halt as soon as the counter exceeds $|x|^2$ steps, so $U$ runs within time $$O(|x|).$$ In this case the corresponding integer is $i_x=1$.

In both cases we have shown that for any input $w$ to $U$, there exists a corresponding positive integer $i_w\in\mathbb{N}_1$ such that $U$ runs for at most $|w|^{i_w+1}$ steps. Consequently, $U$ is a nondeterministic Turing machine that runs in time $O(n^k)$ for any $k\in\mathbb{N}_1$. By Lemma \ref{lemma2.1}, there exists a single-tape nondeterministic Turing machine $U'$ equivalent to $U$ that runs within time $$O(S(n)^2)=O(n^{2k})$$ for any $k\in\mathbb{N}_1$.
\end{proof}

\vskip 0.2 cm
The following theorem is the main upper-bound result of this section:

\begin{theoremsection}
\label{theorem5.2}
The language $L_s$ is in $\mathcal{NP}$, where $L_s$ is accepted by the universal nondeterministic Turing machine $U$ constructed in the proof of Theorem \ref{theorem4.1}.
\end{theoremsection}

\begin{proof}
We first define the family of languages $$\left\{L_s^i\right\}_{i\in\mathbb{N}_1}$$ by
\begin{align*}
L_s^i\overset{{\rm def}}{=}&\text{ the language accepted by $U$ running within time $O(n^i)$ for a fixed $i\in\mathbb{N}_1$.}\\
&\text{ That is, $U$ turns itself off mandatorily when the number of moves it makes }\\
&\text{ during the computation exceeds $n^{i+1}$ steps.}
\end{align*}

Note that the above definition of language $L_s^i$ technically can be realized by adding a new tape to $U$ as a counter to count up to $$n^{i+1}$$ for a fixed $i\in\mathbb{N}_1$. (An alternative construction of languages $L_s^{'i}$ satisfying $L_s^{'i}\subseteq L_s^i$ that does not require an additional counter tape is given in Appendix \ref{sec:a_new_construction_of_lsi}.) With the extra tape, $U$ is forced to halt when either the counter on tape $3$ exceeds $$n^{k+1}$$ or the counter on the newly added tape exceeds $$n^{i+1}. $$ Obviously, for each $i\in\mathbb{N}_1$, $L_s^i$ is a truncation of $L_s$. 

By the construction of $U$ (see Theorem \ref{theorem5.1} above), for any input $w$ to $U$ there exists a corresponding positive integer $i_w\in\mathbb{N}_1$ such that $U$ runs for at most $|w|^{i_w+1}$ steps. Consequently, $$L_s=\bigcup_{i\in\mathbb{N}_1}L_s^i.\eqno(5.1)$$

Moreover, $$L_s^i\subseteq L_s^{i+1},\quad\text{for each fixed $i\in\mathbb{N}_1$}, $$ because any word $w\in L_s^i$ that is accepted by $U$ within $O(n^i)$ steps is certainly accepted by $U$ within $O(n^{i+1})$ steps, i.e., $$w\in L_s^{i+1}. $$ This yields that for any fixed $i\in\mathbb{N}_1$, $$L_s^1\subseteq L_s^2\subseteq\cdots\subseteq L_s^i\subseteq L_s^{i+1}\subseteq\cdots\eqno(5.2) $$

Furthermore, for any fixed $i\in\mathbb{N}_1$, the language $L_s^i$ is accepted by the nondeterministic Turing machine $U$ within time $O(n^i)$ (i.e., in at most $n^{i+1}$ steps). Hence $$L_s^i\in{\rm NTIME}[n^i]\subseteq\mathcal{NP},\quad\text{for any fixed $i\in\mathbb{N}_1$}.\eqno(5.3)$$

Now, combining (5.1), (5.2) and (5.3) immediately yields $$L_s\in \mathcal{NP},$$ as required.
\end{proof}

\vskip 0.3 cm
\noindent{\em Proof of Theorem \ref{theorem1}.} Theorem \ref{theorem1} follows at once from Theorem \ref{theorem4.1} and Theorem \ref{theorem5.2}. This completes the proof.\footnote{Readers may also interpret our result as establishing that $$(\mathcal{NP}\setminus\mathcal{P})\neq{\rm co}(\mathcal{NP}\setminus\mathcal{P}).$$} \Q.E.D

\vskip 0.3 cm
\begin{remark}
By Lemma \ref{lemma2.1} and Theorem \ref{theorem5.2}, the single-tape nondeterministic Turing machine $U_1$ that is equivalent to the nondeterministic Turing machine $U$ (i.e., $L(U_1)=L(U)$) also appears in the enumeration $e$. Likewise, when the machine is required to operate under the condition ``there exists no accepting path", the corresponding single-tape ${\rm co}\mathcal{NP}$ machine $U_1$ appears in the enumeration $e$ as well. More precisely, the proof of Theorem \ref{theorem5.2} shows that for any $i\in\mathbb{N}_1$ one can force the multi-tape universal nondeterministic Turing machine $U$ to run in at most $|w|^{i+1}$ steps on every input $w\in\{0,1\}^*$ (i.e., $U$ runs within time $O(n^{i+1})$). By Lemma \ref{lemma2.1}, such a machine is equivalent to a single-tape nondeterministic Turing machine $U_i$ that runs in time $O(n^{2(i+1)})$, so $L_s^i=L(U_i)$. The single-tape nondeterministic Turing machine $(U_i,2(i+1))$ with time bound $n^{2(i+1)}+2(i+1)$ therefore appears in the enumeration $e$. Nevertheless $L_s\ne\overline{L}(U)$ (note that $L_s=L(U)$). Indeed, for any input $w\in\{0,1\}^*$, one has $$w\in L_s\quad\text{ if and only if }\quad w\not\in\overline{L}(U).$$ 
\end{remark}

\vskip 0.3 cm

We close this section with the following corollary, which follows from the proof of Theorem \ref{theorem5.2}:
\begin{corollarysection}
For any fixed $k\in\mathbb{N}_1$, it holds that 
$$
\bigcup_{i\leq k}L_s^i=L_s^k\in\mathcal{NP},\eqno(5.4) 
$$
where $i\in\mathbb{N}_1$.
\end{corollarysection}
\begin{proof}
It follows clearly from the relations (5.2) and (5.3).
\end{proof}

\begin{remark}
Theorem \ref{theorem5.1} asserts that for any input $w\in\{0,1\}^*$ there exists a positive integer $i_w$ such that $U$ runs for at most $|w|^{i_w+1}$ steps. Combined with (5.4) this yields more transparently $$L_s=\bigcup_{n\in\mathbb{N}_1}L_s^n.\footnote{ Informally and intuitively one may regard the equality as $$\lim_{n\rightarrow \infty}L_s^n=L_s,$$ i.e., for all $\epsilon>0$ there exists an $N_0>0$,  for all $n>N_0$,  $$\#(L_s\setminus L_s^n)<\epsilon,$$where $\#(A)$ denotes the cardinality of the set $A$ and $\setminus$ denotes the set-theoretic minus. Together with the fact that $L_s^n\in\mathcal{NP}$ for all $n\in\mathbb{N}_1$, this gives an intuitive reason why $L_s\in \mathcal{NP}$. The present footnote is not intended as a rigorous mathematical argument (since the inequality $$\#(L_s\setminus L_s^n)<\epsilon$$ holds for all $\epsilon>0$ only when $L_s^n=L_s$, while for other indices one probably have $\#(L_s\setminus L_s^m)=+\infty$); it merely supplies a more intuitive understanding of the conclusion $L_s\in\mathcal{NP}$.}$$

Namely, (5.1) can therefore be deduced from Theorem \ref{theorem5.1} together with (5.4). 
\end{remark}

\vskip 0.3 cm
\section{Breaking the ``Relativization Barrier"}
\label{sec:beating_the_relativization_barrier}
\vskip 0.3 cm

The computation model we use in this section is the {\em query machines}, or the {\em oracle Turing machines}, which is an extension of the multi-tape Turing machine, i.e., Turing machines that are given access to a black box or ``oracle" that can magically solve the decision problem for some language $$X\subseteq\{0,1\}^*. $$ The machine has a special {\em oracle tape} on which it can write a string $$w\in\{0,1\}^*$$ and in one step gets an answer to a query of the form $$\text{``Is $w$ in $X$?"}, $$ which can be repeated arbitrarily often with different queries. If $X$ is a difficult language (say, one that cannot be decided in polynomial time or is even undecidable), then this oracle gives the Turing machine additional power. We first quote its formal definition as follows:

\begin{definition}[cf. the notion of deterministic oracle Turing machines in \cite{AB09}]
\label{definition6.1}
A {\em nondeterministic oracle Turing machine} is a nondeterministic Turing machine $M$ that has a special read-write tape we call $M$'s {\em oracle tape} and three special states $q_{query}$, $q_{yes}$, and $q_{no}$. To execute $M$, we specify in addition to the input a language $X\subseteq\{0,1\}^*$ that is used as the {\em oracle} for $M$. Whenever during the execution $M$ enters the state $q_{query}$, the machine moves into the state $q_{yes}$ if $w\in X$ and $q_{no}$ if $w\not\in X$, where $w$ denotes the contents of the special oracle tape. Note that, regardless of the choice of $X$, a membership query to $X$ counts only as a single computation step. If $M$ is an oracle machine, $X\subseteq\{0,1\}^*$ a language, and $x\in\{0,1\}^*$, then we denote the output of $M$ on input $x$ and with oracle $X$ by $M^X(x)$. An input $x$ is said to be accepted by a nondeterministic oracle Turing machine $M^X$ if there is a computation path of $M^X$ on input $x$ leading to the accepting state of $M^X$.
\end{definition}

Next, we define polynomial-time nondeterministic oracle Turing machines.

\begin{definition}
\label{definition6.2}
Formally, a polynomial-time nondeterministic oracle Turing machine $M^O$ with oracle $O$ is a nondeterministic oracle Turing machine with oracle $O$ such that there exists $k\in\mathbb{N}_1$ for which, for all input $w$ of length $|w|$ (where $|w|\in\mathbb{N}$), $M^O(w)$ halts within $t(|w|)=|w|^k+k$ steps. 
\end{definition}

Similarly, by default, a word $w$ is accepted by a nondeterministic $t(n)$ time-bounded oracle Turing machine $M^X$ with oracle $X$ if there exists at least one computation path $\gamma$ that puts the machine into an accepting state. This is the ``there exists at least one accepting path" condition for nondeterministic oracle Turing machines, on the basis of which the complexity class ${\rm NTIME}^X[t(n)]$ is defined:

\begin{definition}[``there exists at least one accepting path" condition]
\label{definition6.3}
The set of languages decided by all nondeterministic oracle Turing machines with oracle $X$ within time $t(n)$ under the condition ``there exists at least one accepting path" is denoted by ${\rm NTIME}^X[t(n)] $. Thus, $$\mathcal{NP}^X=\bigcup_{k\in\mathbb{N}_1}{\rm NTIME}^X[n^k]. $$
\end{definition}

There is also another condition (namely, there exists no computation path $\gamma$ putting the machine into an accepting state) for nondeterministic oracle Turing machines; this is the dual of the ``there exists at least one accepting path" condition and is the basis on which the notion of an oracle ${\rm co}\mathcal{NP}$ machine is defined. Namely, we have the following: 

\begin{definition}
\label{definition6.4}
A ${\rm co}\mathcal{NP}^A$ machine $M^A$ with oracle $A$ is precisely a polynomial-time nondeterministic oracle Turing machine $M^A$ with oracle $A$ working under the condition ``there exists no accepting path," and we define the language accepted by the ${\rm co}\mathcal{NP}^A$ machine $M^A$ to be $\overline{L}(M^A)$. That is, for any input $w$, $w\in\overline{L}(M^A)$ if and only if the polynomial-time nondeterministic oracle Turing machine $M^A$ rejects $w$.
\end{definition}

Similarly, for a ${\rm co}\mathcal{NP}^A$ machine $M^A$ with oracle $A$, we define the characteristic function $$\chi_{\overline{L}(M^A)}:\{0,1\}^*\rightarrow\{0,1\}$$ 
as follows:
$$
\chi_{\overline{L}(M^A)}(w)=\left\{
                                 \begin{array}{ll}
                                   1, & \hbox{$w\in\overline{L}(M^A)$;} \\
                                   0, & \hbox{$w\notin\overline{L}(M^A)$.}
                                 \end{array}
                               \right.
$$

Let $M^A$ be a polynomial-time nondeterministic oracle Turing machine with oracle $A$. Since $M^A$ can be regarded both as a polynomial-time nondeterministic oracle Turing machine working under the condition ``there exists at least one accepting path" and as a ${\rm co}\mathcal{NP}^A$ machine working under the condition ``there exists no accepting path," we strictly distinguish, in what follows, these two situations in the following. Namely, we use the notation $$M^A(w)=1\quad\text{(respectively, $M^A(w)=0$)} $$ to mean that the polynomial-time nondeterministic oracle Turing machine $M^A$ accepts (respectively, rejects) the input $w$; we use the notation $$\chi_{\overline{L}(M^A)}(w)=1\quad\text{ (respectively, $\chi_{\overline{L}(M^A)}(w)=0$)} $$ to mean that the ${\rm co}\mathcal{NP}^A$ machine $M^A$ accepts (respectively, rejects) the input $w$.

It is clear that for all $w\in\{0,1\}^*$, $w\in\overline{L}(M^A)$ if and only if $$M^A(w)=0,\quad\text{(i.e., the polynomial-time nondeterministic oracle Turing machine $M^A$ rejects $w$)} $$ and $w\notin\overline{L}(M^A)$ if and only if $$M^A(w)=1,\quad\text{(i.e., the polynomial-time nondeterministic oracle Turing machine $M^A$ accepts $w$).} $$

In a manner similar to the definition of the complexity class ${\rm coNTIME}[t(n)]$ in Section \ref{sec:preliminaries}, we can define the complexity class ${\rm coNTIME}^X[t(n)]$, based on the ``there exists no accepting path" condition, as follows.

\begin{definition}[``there exists no accepting path" condition]
\label{definition6.5}
The set of languages decided by all nondeterministic oracle Turing machines with oracle $X$ within time $t(n)$ under the condition ``there exists no accepting path" is denoted by ${\rm coNTIME}^X[t(n)]$, that is, a $t(n)$ time-bounded ${\rm co}\mathcal{NP}^X$ machine with oracle $X$ accepts a string $w$ if and only if there exists no accepting path on input $w$. Thus, $${\rm co}\mathcal{NP}^X=\bigcup_{k\in\mathbb{N}_1}{\rm coNTIME}^X[n^k]. $$
\end{definition}

In what follows we will assume the notion of ${\rm co}\mathcal{P}^X$ machines, which can be defined analogously to Definition \ref{definition6.4} and is therefore omitted here.

In $1975$, Baker, Gill, and Solovay \cite{BGS75} presented a proof of that: $$\text{There is an oracle $A$ for which } \mathcal{P}^A=\mathcal{NP}^A. $$ Baker, Gill, and Solovay \cite{BGS75} suggested that their results imply that ordinary simulation and diagonalization techniques are incapable of proving $\mathcal{P}\neq\mathcal{NP}$. Likewise, if $\mathcal{NP}^A={\rm co}\mathcal{NP}^A$, then Baker, Gill, and Solovay's result also indirectly suggests that ordinary simulation techniques are incapable of proving $\mathcal{NP}\neq{\rm co}\mathcal{NP}$. 

To explore what lies behind this kind of mystery, let us examine two cases: one in which $\mathcal{P}^O=\mathcal{NP}^O={\rm co}\mathcal{NP}^O$ for some oracle $O$, and another in which $\mathcal{P}^A\ne\mathcal{NP}^A={\rm co}\mathcal{NP}^A$ for some oracle $A$.

Note that the result of Baker, Gill, and Solovay \cite{BGS75} also implies that for the same oracle $O$, $$\mathcal{NP}^O={\rm co}\mathcal{NP}^O, $$ because if $\mathcal{P}^O=\mathcal{NP}^O$, then ${\rm co}\mathcal{P}^O={\rm co}\mathcal{NP}^O$. Further, since $$\mathcal{P}^O={\rm co}\mathcal{P}^O, $$we have $$\mathcal{NP}^O={\rm co}\mathcal{NP}^O,$$ which belongs to the first case mentioned above: $$\mathcal{P}^O=\mathcal{NP}^O={\rm co}\mathcal{NP}^O. $$ It should be pointed out that, in this scenario, we are unable to apply simulation techniques to separate $\mathcal{NP}^O$ from ${\rm co}\mathcal{NP}^O$; the reasons will be given in Theorem \ref{theorem6.1} below.

\begin{theoremsection}
\label{theorem6.1}
If $\mathcal{P}^O=\mathcal{NP}^O={\rm co}\mathcal{NP}^O$ and we assume that the set of all ${\rm co}\mathcal{P}^O$ machines is enumerable, then simulation techniques are unable to create a language $L_s^O\notin{\rm co}\mathcal{NP}^O$ and are therefore unable to separate $\mathcal{NP}^O$ from ${\rm co}\mathcal{NP}^O$.
\end{theoremsection}

\begin{proof}
First note that $\mathcal{P}^O=\mathcal{NP}^O$; thus, ${\rm co}\mathcal{P}^O={\rm co}\mathcal{NP}^O$. Clearly, in this scenario,  separating $\mathcal{NP}^O$ from ${\rm co}\mathcal{NP}^O$ is equivalent to separating $\mathcal{P}^O$ from ${\rm co}\mathcal{P}^O$. Moreover, if the set of all ${\rm co}\mathcal{P}^O$ machines is not enumerable, then of course we cannot apply simulation techniques. We therefore assume that the set of ${\rm co}\mathcal{P}^O$ machines is enumerable, as stated in the premise of the theorem. Since we have assumed that the set of all ${\rm co}\mathcal{P}^O$ machines is enumerable, the rest of the proof is similar to Remark \ref{remark4.1}. Namely, the simulation technique is unable to create a language not in ${\rm co}\mathcal{P}^O$ via a universal deterministic oracle Turing machine $U^O$, because a polynomial-time deterministic oracle Turing machine that works under the ``there exists no accepting path" condition is still a polynomial-time deterministic oracle Turing machine. In short, the ``there exists at least one accepting path" condition and the ``there exists no accepting path" condition for deterministic oracle Turing machines are symmetrical. Specifically, for an arbitrary deterministic oracle Turing machine $M^O$, if $M^O$ accepts the input $w$, then there is only one accepting path, and if $M^O$ rejects the input $w$, then there is also only one rejecting path. This differs from the situation for nondeterministic oracle Turing machines. If we apply the same techniques (i.e., simulation techniques) to create a language not in ${\rm co}\mathcal{P}^O$, then the universal deterministic oracle Turing machine $U^O$ (which runs within time $O(n^k)$ for any $k\in\mathbb{N}_1$) that must be constructed accepts the language $L(U^O)\in\mathcal{P}^O$ (the proof is similar to that of Theorem \ref{theorem6.4} below). But by the symmetrical properties of the ``there exists at least one accepting path" condition and the ``there exists no accepting path" condition with respect to deterministic oracle Turing machines, we clearly have $L(U^O)\in{\rm co}\mathcal{P}^O$. Thus we are unable to claim that $L(U^O)\notin{\rm co}\mathcal{P}^O$.
\end{proof}

\vskip 0.2cm
We now assume another scenario, namely the case in which there exists an oracle $A$ such that $$\mathcal{P}^A\ne\mathcal{NP}^A={\rm co}\mathcal{NP}^A.$$ If the above case holds, let us explore what lies behind this kind of mystery.

First of all, for convenience, let us denote by $\mathcal{T}^X$ the set of all ${\rm co}\mathcal{NP}^X$ machines, for any oracle $X$. Moreover, {\em it suffices to consider $\mathcal{T}^X$ to be the set of all two-tape oracle ${\rm co}\mathcal{NP}^X$ machines: one is the input tape,\footnote{The input tape also functions as the working tape.} and the other is the oracle tape.} This is because, in a manner similar to the proof of Lemma \ref{lemma2.1}, it can be shown that $t(n)$ time-bounded $k$-tape oracle nondeterministic Turing machines (with an additional oracle tape) are equivalent to $t^2(n)$ time-bounded single-tape oracle nondeterministic Turing machines (with an additional oracle tape). That is, we merge the input tape and the working tapes into a single tape while leaving the oracle tape unchanged.

Suppose that $\mathcal{P}^A\ne\mathcal{NP}^A={\rm co}\mathcal{NP}^A$. In this case, to break the ``Relativization Barrier" or to explore the mystery behind the implications of Baker, Gill, and Solovay's result \cite{BGS75}, let us make the following rational assumptions, which are intended for comparison with the general point of view concerning certain properties of Turing machines that must be satisfied when applying ``diagonalization" techniques (see e.g., p. 73, \cite{AB09}):
\begin{assumptions}
\label{assumptions}
Based on the above, we can make the following specific assumptions:
\begin{enumerate}
  \item {The oracle ${\rm co}\mathcal{NP}$ machines (by Definition \ref{definition6.4}, equivalently the polynomial-time nondeterministic oracle Turing machines) can be encoded as strings over $\{0,1\}$;}
  \item {There exist universal nondeterministic oracle Turing machines that can simulate any other polynomial-time nondeterministic oracle Turing machine, in a manner similar to the behavior of the universal nondeterministic Turing machine $U$ constructed in Section \ref{sec:simulation_of_all_polynomial-time_nondeterministic_turing_machines};}
  \item {The simulation can be performed within time 
  $$
  O(T(n)\log T(n)),
  $$where $T(n)$ is the time complexity of the simulated polynomial-time nondeterministic oracle Turing machine.}
\end{enumerate}
\end{assumptions}

We will then prove the following interesting theorems:

\begin{theoremsection}
\label{theorem6.2}
If $\mathcal{P}^A\ne\mathcal{NP}^A$ and we suppose that the set $\mathcal{T}^A$ of all ${\rm co}\mathcal{NP}^A$ machines with oracle $A$ is enumerable and that the enumeration is $e'$, and if we have the above rational assumptions (i.e., Assumption \ref{assumptions}), then there exists a language $L_s^A$ accepted by a universal nondeterministic oracle Turing machine $U^A$ with oracle $A$ but not by any ${\rm co}\mathcal{NP}^A$ machines with oracle $A$. Namely, there is a language $L_s^A$ such that $$L_s^A\notin{\rm co}\mathcal{NP}^A.$$
\end{theoremsection}
\begin{proof}
Since the set $\mathcal{T}^A$ of all (two-tape) ${\rm co}\mathcal{NP}^A$ machines is assumed to be enumerable, and since we further have the above rational assumptions --- i.e., (1) any ${\rm co}\mathcal{NP}^A$ machine can be encoded as a string over $\{0,1\}$; (2) there exist universal nondeterministic oracle Turing machines with oracle $A$ that can simulate any other polynomial-time nondeterministic oracle Turing machine with oracle $A$, in a manner similar to the behavior of the universal nondeterministic Turing machine $U$ constructed in Section \ref{sec:simulation_of_all_polynomial-time_nondeterministic_turing_machines}; (3) the simulation can be done within time $$O(T(n)\log T(n)), $$ where $T(n)$ is the time complexity of the simulated polynomial-time nondeterministic oracle Turing machine with oracle $A$ --- the remainder of the proof is to show that there exists a universal nondeterministic oracle Turing machine $U^A$ with oracle $A$ accepting the language $$L_s^A\notin{\rm co}\mathcal{NP}^A$$ in a manner similar to that of Theorem \ref{theorem4.1}. For the convenience of the reader and for clarity, we proceed as follows: Let $U^A$ be a five-tape universal nondeterministic oracle Turing machine that operates as follows on an input string $x$ of length $n$:
\begin{enumerate}
\item{ $U^A$ decodes the tuple encoded by $x$. If $x$ is not the encoding of some ${\rm co}\mathcal{NP}^A$ machine $M^A_j$ with oracle $A$ for some $j$ (by Definition \ref{definition6.4}, equivalently, if $x$ cannot be decoded into some polynomial-time nondeterministic oracle Turing machine $M^A_j$ for some $j$, i.e., if $x$ represents the trivial oracle Turing machine with an empty next-move function), then GOTO $6$; else determine $t$, the number of tape symbols used by $M^A_j$; $s$, its number of states; and $k$, its order.\footnote {We suppose that the order $10^k1$ of $M^A_j$ for some $k\in\mathbb{N}_1$ is also encoded into the binary string representing ${\rm co}\mathcal{NP}^A$ machine $M^A_j$, which is similar to Section \ref{sec:all_polynomial_time_nondeterministic_turing_machines_are_enumerable}.} The third tape of $U^A$ can be used as ``scratch" memory to calculate $t$.}
\item{ $U^A$ then lays off on its second tape $n$ blocks of $$\lceil\log t\rceil$$ cells each, the blocks being separated by a single cell holding a marker $\#$, i.e., there are $$(1+\lceil\log t\rceil)n$$ cells in all. Each tape symbol occurring in a cell of $M^A_j$'s tape will be encoded as a binary number in the corresponding block of the second tape of $U^A$. Initially, $U^A$ places $M^A_j$'s input, in binary-coded form, in the blocks of tape $2$, filling the unused blocks with the code for the blank.}
\item{ On tape $3$, $U^A$ sets up a block of $$\lceil(k+1)\log n\rceil$$ cells, initialized to all $0$'s. Tape $3$ is used as a counter to count up to $$n^{k+1}. $$
        }
\item{ On tape $4$, $U^A$ reads and writes the contents of the oracle tape of $M^A_j$. That is, tape $4$ is the oracle tape of $U^A$, which is used to simulate the oracle tape of $M^A_j$.
    }
\item{Using nondeterminism, $U^A$ simulates $M^A_j$, using tape $1$ (its input tape) to determine the moves of $M^A_j$ and tape $2$ to simulate the tape of $M^A_j$. The moves of $U^A$ are counted in binary in the block of tape $3$, and tape $5$ is used to hold the states of $M^A_j$. If the counter on tape $3$ overflows, $U^A$ halts without accepting. The specified simulation is the following:

 ``On input $x=1^{n-m}\langle M^A_j\rangle$ where $m=|\langle M^A_j\rangle|$, $U^A$ on input $1^{n-m}\langle M^A_j\rangle$ simulates $M^A_j$ on input $1^{n-m}\langle M^A_j\rangle$ using nondeterminism in $$n^{k+1}  $$ time and outputs its answer, i.e., accepting if $M^A_j$ accepts and rejecting if $M^A_j$ rejects." 

  Also note that, although we say that the input $x$ is decoded into a ${\rm co}\mathcal{NP}^A$ machine $M^A_j$, $U^A$ in fact simulates the $n^k+k$ time-bounded nondeterministic oracle Turing machine $M^A_j$ with oracle $A$. That is, $U^A$ accepts the input $x=1^{n-m}\langle M^A_j\rangle$ if there exists at least one accepting path of $M^A_j$ on input $x$, and $U^A$ rejects the input $x=1^{n-m}\langle M^A_j\rangle$ if there exists no accepting path of $M^A_j$ on input $x$. 
  }
 \item{ Since $x$ is an encoding of the trivial oracle Turing machine with an empty next-move function, $U^A$ sets up a block of $$\lceil 2\times\log n\rceil$$ cells on tape $3$, initialized to all $0$'s. Tape $3$ is used as a counter to count up to $$n^2.$$ Using its nondeterministic choices, $U^A$ moves according to the path given by $x$. The moves of $U^A$ are counted in binary in the block of tape $3$. If the counter on tape $3$ overflows, then $U^A$ halts. $U^A$ accepts $x$ if and only if there is at least one computation path from the start state of $U^A$ leading to the accept state and the total number of moves does not exceed $$n^2$$ steps, so the computation lies within $$O(n)$$ steps. Note that the factor $2$ in $$\lceil 2\times\log n\rceil$$ is fixed; it is a default setting.}
\end{enumerate}

The nondeterministic oracle Turing machine $U^A$ with oracle $A$ constructed above is of time complexity, say $S$, which is currently unknown. It, of course, accepts some language $L^A_s$.

We claim that $L^A_s\notin{\rm co}\mathcal{NP}^A$. Indeed, suppose for the sake of contradiction that $L^A_s$ is decided by some $n^k+k$ time-bounded ${\rm co}\mathcal{NP}^A$ machine $M^A_i$. Then, {\em by the definition of $\mathcal{T}^A$}, $M^A_i$ is a two-tape (with one input tape and one oracle tape) $T(n)=n^k+k$ time-bounded ${\rm co}\mathcal{NP}^A$ machine. By Definition \ref{definition6.4}, we may regard it as an $n^k+k$ time-bounded nondeterministic oracle Turing machine with oracle $A$, and further let $U^A$ on input $1^{n-m}\langle M^A_i\rangle$ simulate $M^A_i$ on input $1^{n-m}\langle M^A_i\rangle$. Suppose $M^A_i$ have $s$ states and $t$ tape symbols, and let $m$ be the length of the shortest encoding of $M^A_i$, i.e., $$m=\min\{|\langle M^A_i\rangle|\,:\,\text{different binary strings $\langle M^A_i\rangle$ represent the same $M^A_i$}\}. $$

Since $M^A_i$ is represented by infinitely many strings and appears infinitely often in the enumeration $e'$, and since $U^A$ was set to count up to $n^{k+1}$ on tape $3$, by Assumption (3), to simulate $$T(n)$$ steps of $M^A_i$, $U^A$ must perform $$T(n)\log T(n)$$ steps. Since

$$\aligned
\lim_{n\rightarrow\infty}&\frac{T(n)\log T(n)}{n^{k+1}}\\
=&\lim_{n\rightarrow\infty}\frac{(n^k+k)\log(n^k+k)}{n^{k+1}}\\
=&\lim_{n\rightarrow\infty}\left(\frac{n^k\log(n^k+k)}{n^{k+1}}+\frac{k\log(n^k+k)}{n^{k+1}}\right)\\
=&0\\
<&1,
\endaligned$$

\noindent there exists an $N'>0$ such that for any $N>N'$, $$T(N)\log T(N)<N^{k+1}, $$ which implies that for a sufficiently long $w$, say $|w|\ge N'$, and $M^A_w$ denoted by such a $w$ is $M^A_i$, we have $$T(|w|)\log T(|w|)<|w|^{k+1}. $$ Thus, on input $w$, $U^A$ has sufficient time to simulate $M^A_w$. This yields  $$1^{|w|-m}\langle M^A_i\rangle \in L^A_s\quad\text{if and only if}\quad 1^{|w|-m}\langle M^A_i\rangle\in L(M^A_w)= L(M^A_i).\eqno(6.1) $$

But by the relation between the languages $L(M^A_i)$ accepted by polynomial-time nondeterministic oracle Turing machine $M^A_i$ with oracle $A$ and the language $\overline{L}(M^A_i)$ accepted by the ${\rm co}\mathcal{NP}^A$ machine $M^A_i$, we have $$1^{|w|-m}\langle M^A_i\rangle\in L(M^A_i)\quad\text{if and only if}\quad 1^{|w|-m}\langle M^A_i\rangle\notin\overline{L}(M^A_i).\eqno(6.2)$$

From (6.1) and (6.2) we conclude that if $U^A$ accepts the input $1^{|w|-m}\langle M^A_i\rangle$, i.e., $$1^{|w|-m}\langle M^A_i\rangle\in L^A_s,$$ then $$1^{|w|-m}\langle M^A_i\rangle\notin\overline{L}(M^A_i), $$ or equivalently, $$\chi_{\overline{L}(M^A_i)}(1^{|w|-m}\langle M^A_i\rangle)=0. $$ Likewise, if $U^A$ rejects the input $1^{|w|-m}\langle M^A_i\rangle $, i.e., $$1^{|w|-m}\langle M^A_i\rangle\notin L^A_s, $$ then $$1^{|w|-m}\langle M^A_i\rangle\in\overline{L}(M^A_i), $$ or equivalently,$$\chi_{\overline{L}(M^A_i)}(1^{|w|-m}\langle M^A_i\rangle)=1, $$ a contradiction with our assumption that $L^A_s$ is decided by the $n^k+k $ time-bounded ${\rm co}\mathcal{NP}^A$ machine $M^A_i$. 

The above arguments further yield that there exists no ${\rm co}\mathcal{NP}^A$ machine $M^A_i$ in the enumeration $e'$ accepting the language $L^A_s$. Since all ${\rm co}\mathcal{NP}^A$ machines appear in the list $e'$, it follows that $$L^A_s\notin{\rm co}\mathcal{NP}^A.$$ This finishes the proof.

\end{proof}

Next we show that the universal nondeterministic oracle Turing machine $U^A$ with oracle $A$ runs within time $O(n^k)$ for all $k\in\mathbb{N}_1$:

\begin{theoremsection}
\label{theorem6.3}
The universal nondeterministic oracle Turing machine $U^A$ with oracle $A$ constructed in the proof of Theorem \ref{theorem6.2} runs within time $O(n^k)$ for any $k\in\mathbb{N}_1$. 
\end{theoremsection}
\begin{proof}
The proof of this theorem is similar to that of Theorem \ref{theorem5.1}. Namely, to prove the theorem it suffices to prove that for any input $w$ to $U^A$, there exists a corresponding positive integer $i_w\in\mathbb{N}_1$ such that $U^A$ runs for at most $$|w|^{i_w+1}$$ 
steps. Since the remainder of the proof is similar to that of Theorem \ref{theorem5.1}, the detail are therefore omitted.
\end{proof}

We are now ready to prove that the language $L^A_s$ is in fact in $\mathcal{NP}^A$:

\begin{theoremsection}
\label{theorem6.4}
The language $L_s^A$ accepted by the nondeterministic oracle Turing machine $U^A$ with oracle $A$ is in fact in $\mathcal{NP}^A$.
\end{theoremsection}
\begin{proof}
The proof of this theorem is similar to that of Theorem \ref{theorem5.2}. Namely, the first step is to define the family of languages: $$\left\{L^{A,i}_s\right\}_{i\in\mathbb{N}_1}. $$ Then, by proof of Theorem \ref{theorem6.3}, we have $$L_s^A=\bigcup_{i\in\mathbb{N}_1}L^{A,i}_s.\eqno(6.3) $$ Next we must show that for all $i\in\mathbb{N}_1$, $$L^{A,i}_s\in\mathcal{NP}^A,\eqno(6.4) $$ and $$L^{A,1}_s\subseteq L^{A,2}_s\subseteq\cdots\subseteq L^{A,i}_s\subseteq L^{A,i+1}_s\subseteq\cdots\eqno(6.5) $$ Finally, the conclusion follows from (6.3), (6.4), and (6.5). Since the remainder of the proof is similar to that of Theorem \ref{theorem5.2}, the more specific details are therefore omitted.
\end{proof}

\vskip 0.1 cm
Combining Theorem \ref{theorem6.2} and Theorem \ref{theorem6.4} actually yields the following result:
\begin{theoremsection}
\label{theorem6.5}
If $\mathcal{P}^A\ne\mathcal{NP}^A$, if we assume that the set $\mathcal{T}^A$ of all ${\rm co}\mathcal{NP}^A$ machines is enumerable, and if we further assume the rational assumptions above (i.e., Assumption \ref{assumptions}), then $$\mathcal{NP}^A\ne{\rm co}\mathcal{NP}^A.$$ \Q.E.D
\end{theoremsection}

\vskip 0.1cm

But we now have the fact that $\mathcal{P}^A\ne\mathcal{NP}^A={\rm co}\mathcal{NP}^A$, so the assumption in Theorem \ref{theorem6.5} is not true, i.e., the set $\mathcal{T}^A$ of all ${\rm co}\mathcal{NP}^A$ machines is not enumerable. Indeed, by $\mathcal{NP}^A={\rm co}\mathcal{NP}^A$, there exists a ${\rm co}\mathcal{NP}^A$ machine accepting the language $L^A_s$ that differs from any languages accepted by ${\rm co}\mathcal{NP}^A$ machines in the enumeration $e'$. Thus the ${\rm co}\mathcal{NP}^A$ machine accepting the language $L_s^A$ is not in the enumeration $e'$, showing that the set $\mathcal{T}^A$ of all ${\rm co}\mathcal{NP}^A$ machines is in fact not enumerable. The truths behind this kind of mystery are therefore expressed by the following theorem:

\begin{theoremsection}
\label{theorem6.6}
Under some rational assumptions (i.e., Assumption \ref{assumptions} above), and if $\mathcal{P}^A\ne\mathcal{NP}^A={\rm co}\mathcal{NP}^A$, then the set $\mathcal{T}^A$ of all ${\rm co}\mathcal{NP}^A$ machines is not enumerable. Thereby, ordinary simulation techniques will generally {\em not} apply to the relativized versions of the $\mathcal{NP}$ versus ${\rm co}\mathcal{NP}$ problem.\Q.E.D
\end{theoremsection}

We are now at the right point to present the proof of Theorem \ref{theorem7} as follows:
\vskip 0.3 cm
{\em Proof of Theorem \ref{theorem7}.} Note that the statement of Theorem \ref{theorem7} is precisely the same as that of Theorem \ref{theorem6.6}; therefore the proof is complete.\Q.E.D

\vskip 0.3 cm
\begin{remark}
\label{remark6.1}
Thus, from Theorem \ref{theorem7}, we know that Baker, Gill, and Solovay's result \cite{BGS75} implies that $\mathcal{NP}\ne{\rm co}\mathcal{NP}$ is not necessarily a necessary and sufficient condition for $\mathcal{NP}^O\ne{\rm co}\mathcal{NP}^O$ for all oracle $O$. Likewise, Yao's conclusion \cite{Yao85} on oracle separation of {\it polynomial hierarchy} also implies that $\mathcal{NP}={\rm co}\mathcal{NP}$ is not necessarily a necessary and sufficient condition for $\mathcal{NP}^O={\rm co}\mathcal{NP}^O$ for all oracle $O$. In other words, the so-called ``Relativization Barrier" is not really an absolute barrier, since an absolute barrier cannot be broken at all, at least from the author's point of view.
\end{remark}

\vskip 0.3 cm
\begin{remark}
From the arguments in this section we can further conclude the following: if $\mathcal{P}^A\ne\mathcal{NP}^A$ and the set $\mathcal{T}^A$ of all ${\rm co}\mathcal{NP}^A$ machines with oracle $A$ is indeed enumerable and $\mathcal{NP}^A={\rm co}\mathcal{NP}^A$, then at least one of the aforementioned conditions (1),(2) and (3) in {\em Assumption} \ref{assumptions} does not hold.
\end{remark}

\vskip 0.3 cm
\section{Rich Structure of ${\rm co}\mathcal{NP}$}
\label{sec:structure_of_coNP}
\vskip 0.3 cm

In complexity theory, or computational complexity, problems that lie in the complexity class $\mathcal{NP}$ but are neither in the class $\mathcal{P}$ nor $\mathcal{NP}$-complete are called $\mathcal{NP}$-intermediate, and the class of such problems is called $\mathcal{NPI}$. The well-known Ladner's theorem, proved in 1975 by Ladner \cite{Lad75}, asserts that if the complexity classes $\mathcal{P}$ and $\mathcal{NP}$ are different, then the class $\mathcal{NPI}$ is not empty. Namely, the complexity class $\mathcal{NP}$ contains problems that are neither in $\mathcal{P}$ nor $\mathcal{NP}$-complete. We call this elegant theorem a reflection of the rich structure of the complexity class $\mathcal{NP}$.

Our main goal in this section is to establish a result saying that, similarly to the complexity class $\mathcal{NP}$ having a rich structure, the complexity class ${\rm co}\mathcal{NP}$ also contains intermediate languages that are neither in $\mathcal{P}$ nor ${\rm co}\mathcal{NP}$-complete. To do so in a simple way, we quote the following useful result whose proof can be found in \cite{Lad75}:

\begin{lemma}[\cite{Lad75}]
\label{lemma7.1}
(Suppose that $\mathcal{P}\ne\mathcal{NP}$). There is a language $L_{\rm inter}\in\mathcal{NP}$ such that $L_{\rm inter}$ is not in $\mathcal{P}$ and $L_{\rm inter}$ is not $\mathcal{NP}$-complete.\Q.E.D
\end{lemma}

Next, we show that the language of all tautologies is ${\rm co}\mathcal{NP}$-complete:

\begin{lemma}
\label{lemma7.2}
Let ${\rm TAUT}=\{\varphi\,:\,\varphi \text{ is a tautology}\}$. Then the language ${\rm TAUT}$ is ${\rm co}\mathcal{NP}$-complete.
\end{lemma}
\begin{proof}
It is sufficient to show that for every language $L\in{\rm co}\mathcal{NP}$, $$L\leq_m^p{\rm TAUT},$$ where $\leq_m^p$ denotes polynomial-time many-one reduction \cite{Kar72}. We can modify the Cook-Levin reduction \cite{Coo71, Lev73} from $\overline{L}=\{0,1\}^*\setminus L$ to ${\rm SAT}$ (recall that all nondeterministic Turing machines are assumed to with input alphabet $\{0,1\}$ because that will be all we need; see Section \ref{sec:all_polynomial_time_nondeterministic_turing_machines_are_enumerable}). Thus, for any input $w\in\{0,1\}^*$, the Cook-Levin reduction produces a formula $\varphi_w$ that is satisfiable if and only if $w\in\overline{L}$. In other words, the formula $\neg\varphi_x$ belongs to ${\rm TAUT}$ if and only if $x\in L$, which completes the proof.
\end{proof}

\vskip 0.3 cm
We are now at the right point to give the proof of Theorem \ref{theorem4} as follows:
\vskip 0.2 cm

\noindent{\em Proof of Theorem \ref{theorem4}.} Let $L=\{0,1\}^*\setminus L_{\rm inter}$, where the language $L_{\rm inter}$ is the same as in Lemma \ref{lemma7.1}. Such a language indeed exists by Corollary \ref{corollary2}, which states that $\mathcal{P}\ne\mathcal{NP}$. 

Then $L\in{\rm co}\mathcal{NP}$, and we will show that $L$ is neither in $\mathcal{P}$ nor ${\rm co}\mathcal{NP}$-complete. To see this, suppose that $L\in\mathcal{P}$. Since $\mathcal{P}={\rm co}\mathcal{P}$, we would then have $\overline{L}=L_{\rm inter}\in\mathcal{P}$, contradicting the fact that $L_{\rm inter}\notin\mathcal{P}$ by Lemma \ref{lemma7.1}. Further assume that $L$ is ${\rm co}\mathcal{NP}$-complete. Then ${\rm TAUT}\leq_m^p L$. By Lemma \ref{lemma7.2}, $\overline{L}=L_{\rm inter}$ would be $\mathcal{NP}$-complete, again a contradiction to the fact that $L_{\rm inter}$ is not $\mathcal{NP}$-complete. Summarizing the two cases above, we have shown that $L$ is ${\rm co}\mathcal{NP}$-intermediate. This finishes the proof. \Q.E.D

\vskip 0.3 cm
\section{Frege Systems}
\label{sec:frege_systems}
\vskip 0.3 cm

The focal point of this section is Frege systems, which are very strong proof systems in the propositional setting \cite{CR79}, based on axiom schemes and rules such as modus ponens. While Frege systems operate with Boolean formulas as lines, the extended Frege system EF works with Boolean circuits; see, e.g., \cite{Jer05}. Furthermore, establishing lower bounds on Frege or even extended Frege systems constitutes a major open problem in proof complexity; see e.g., \cite{BP01}.

Informally, a proof system for a language $\mathcal{L}$ is a definition of what is considered to be a proof that $\Phi\in\mathcal{L}$; see e.g., \cite{CR79}. The key features of a proof system are that it is sound (i.e., only formulas in $\mathcal{L}$ have proofs), complete (i.e., all formulas in $\mathcal{L}$ have proofs), and that there exists an algorithm with running time polynomial in $|\pi|$ to check whether $\pi$ is a proof that $\Phi\in\mathcal{L}$ \cite{BH19}.

Let $\mathfrak{F}$ be the set of functions $f:\Sigma_1^*\rightarrow\Sigma_2^*$, where $\Sigma_1$ and $\Sigma_2$ are any finite alphabets, such that $f$ can be computed by a deterministic Turing machine in time bounded by a polynomial in the length of the input. We then have the following definition:

\begin{definition}[\cite{CR79}]
\label{definition8.1}
If $L\subseteq\Sigma^*$, a {\em proof system} for $L$ is a function $f:\Sigma_1^*\rightarrow L$ for some alphabet $\Sigma_1$ and $f\in \mathfrak{F}$ such that $f$ is onto. We say that the proof system is {\em polynomially bounded} iff there is a polynomial $p(n)$ such that for all $y\in L$ there is an $x\in\Sigma_1^*$ such that $y=f(x)$ and $|x|\leq p(|y|)$, where $|z|$ denotes the length of a string $z$.
\end{definition}

If $y=f(x)$, then $x$ is a {\em proof} of $y$, and $x$ is a {\em short} proof of $y$ if, in addition, $|x|\leq p(|y|)$. Thus a proof system $f$ is polynomially bounded iff there is a bounding polynomial $p(n)$ with respect to which every $y\in L$ has a short proof.

In particular, Frege proof systems are proof systems for propositional logic. As a matter of fact, Frege systems are the usual ``textbook" proof systems for propositional logic based on axioms and rules; see, e.g., \cite{CR79}. A Frege system composed of a finite set of axiom schemes and a finite number of rules is a possible axiom scheme. A Frege proof is a sequence of formulas where each formula is either a substitution instance of an axiom or can be inferred from previous formulas by a valid inference rule. At the same time, Frege systems are required to be sound and complete. However, the exact choice of the axiom schemes and rules does not matter, as any Frege systems are polynomially equivalent; see e.g., \cite{CR79} or Theorem \ref{theorem8.1} below. Thus we may assume without loss of generality that modus ponens (see e.g., Remark \ref{remark8.1} below) is the only rule of inference.

\begin{definition}[\cite{BG98}]
\label{definition8.2}
A Frege proof system $F$ is an inference system for propositional logic based on
\begin{enumerate}
  \item [{\em (1)}]{ a language of well-formed formulas obtained from a numerable set of propositional variables and any finite propositionally complete set of connectives;}
  \item [{\em (2)}]{a finite set of axiom schemes;}
  \item [{\em (3)}]{and the rule of {\em Modus Ponens}
  $$
  \frac{A\quad A\rightarrow B}{B}.
  $$
  }
\end{enumerate}
A proof $P$ of the formula $A$ in a Frege system is a sequence $A_1,\cdots, A_n$ of formulas such that $A_n$ is $A$ and every $A_i$ is either an {\em instance} of an axiom scheme or it is obtained by an application of the Modus Ponens from premises $A_j$ and $A_k$ with $j,k<i$. 
\end{definition}

\vskip 0.3 cm
\begin{remark}
\label{remark8.1}
The well-known inference rule of Modus Ponens is its only rule of inference: $$\frac{A\quad A\rightarrow B}{B}$$ We give a more intuitive explanation about the inference rule of Modus Ponens: For example, in the proof of Theorem \ref{theorem5.2} in Section \ref{sec:showing_l_s_in_np}, where we prove that $L_s\in\mathcal{NP}$, $A$ is defined by $$A\overset{\rm def}{=}(5.1)\bigwedge(5.2)\bigwedge(5.3),$$ $B$ represents the proposition $L_s\in\mathcal{NP}$, namely $$B\overset{\rm def}{=}L_s\in\mathcal{NP}, $$ and one of the main deductions of Theorem \ref{theorem5.2}, besides showing that $A$ is true, is to show that $A$ implies $B$.
\end{remark}

The axiom schemes of a Frege proof system are the following (see e.g., \cite{Bus99}):
$$\aligned
&(P\wedge Q)\rightarrow P\\
&(P\wedge Q)\rightarrow Q\\
&P\rightarrow(P\vee Q)\\
&Q\rightarrow (P\vee Q)\\
&(P\rightarrow Q)\rightarrow((P\rightarrow\neg Q)\rightarrow\neg P)\\
&(\neg\neg P)\rightarrow P\\
&P\rightarrow(Q\rightarrow P\wedge Q)\\
&(P\rightarrow R)\rightarrow((Q\rightarrow R)\rightarrow(P\vee Q\rightarrow R))\\
&P\rightarrow(Q\rightarrow P)\\
&(P\rightarrow Q)\rightarrow(P\rightarrow(Q\rightarrow R))\rightarrow(P\rightarrow R).
\endaligned$$

More generally, a Frege system is specified by any finite complete set of Boolean connectives and a finite set of axiom schemes and rule schemes, provided it is implicationally sound and implicationally complete.

\begin{definition}
\label{definition8.3}
The length of a Frege proof is the number of symbols in the proof. The length $|\varphi|$ of a formula $\varphi$ is the number of symbols in $\varphi$.
\end{definition}

We now introduce the notion of $p$-simulation between two proof systems.

\begin{definition}[\cite{CR79}]
\label{definition8.4}
If $f_1:\Sigma_1^*\rightarrow L$ and $f_2: \Sigma_2^*\rightarrow L$ are proof systems for $L$, then $f_2$ {\em $p$-simulates} $f_1$ provided there is a function $g:\Sigma_1^*\rightarrow \Sigma_2^*$ such that $g\in\mathfrak{F}$, and $f_2(g(x))=f_1(x)$ for all $x$.
\end{definition}

The notion of an abstract propositional proof system is given as follows:

\begin{definition}[\cite{Bus99}]
\label{definition8.5}
An {\em abstract propositional proof system} is a polynomial-time computable function $g$ such that $Range(g)=TAUT$; i.e., the range of $g$ is the set of all Boolean tautologies. A $g$-proof of a formula $\phi$ is a string $w$ such that $g(w)=\phi$.
\end{definition}

\vskip 0.3 cm
\subsection{Lower Bounds for Frege Systems}

Let $\mathcal{F}_1$ and $\mathcal{F}_2$ be two arbitrary Frege systems. The following theorem indicates that these standard proof systems for the propositional calculus are about equally powerful:
\begin{theoremsection}[\cite{Rec76,CR79}]
\label{theorem8.1}
Any two Frege systems $p$-simulate each other. Hence, one Frege system is polynomially bounded if and only if all Frege systems are.\Q.E.D
\end{theoremsection}

The following theorem gives a necessary and sufficient condition for the complexity class $\mathcal{NP}$ to be closed under complement. For completeness, the proof is also quoted from \cite{CR79}:
\begin{theoremsection}[\cite{CR79}]
\label{theorem8.2}
$\mathcal{NP}$ is closed under complementation if and only if $TAUT$ is in $\mathcal{NP}$.
\end{theoremsection}
\begin{proof}
(``If" direction.) Suppose that the set of tautologies is in $\mathcal{NP}$. Then every set $L$ in $\mathcal{NP}$ is reducible to the complement of the tautologies \cite{Coo71}, i.e., there is a function $f$ that can be computed by a deterministic Turing machine in time bounded by a polynomial in the length of the input such that for all strings $x$, $x\in L$ iff $f(x)$ is not a tautology. Hence a nondeterministic procedure for accepting the complement of $L$ is: on input $x$, compute $f(x)$, and accept $x$ if $f(x)$ is a tautology, using the nondeterministic algorithm for tautologies assumed above. Therefore the complement of $L$ is in $\mathcal{NP}$.

(``Only if" direction.) Suppose that $\mathcal{NP}$ is closed under complementation. Then the complement of the set of tautologies is in $\mathcal{NP}$, since to verify that a formula is not a tautology, one can guess a truth assignment and verify that it falsifies the formula.
\end{proof}

Since nondeterministic Turing machines can simulate Frege proofs (see e.g., \cite{Bus99}), and by Theorem \ref{theorem8.1} any two Frege systems $p$-simulate each other, we can present the proof of Theorem \ref{theorem6} as follows:\\

\skip 0.6 cm
\noindent{\em Proof of Theorem \ref{theorem6}.} We prove the theorem by contradiction. Suppose that there is a polynomial $p(n)$ such that for all $\psi\in TAUT$, there exists a proof $\oint$ of $\psi$ of length at most $p(|\psi|)$ (see Definition \ref{definition8.3}). Then a nondeterministic Turing machine $M$ on input $\psi$ can guess $\oint$ of length at most $p(|\psi|)$ and check whether $\oint$ is correct in deterministic polynomial time. Since the set of tautologies is ${\rm co}\mathcal{NP}$-complete and the above nondeterministic machine $M$ witnesses that $TAUT$ is in $\mathcal{NP}$, Theorem \ref{theorem8.2} then implies $$\mathcal{NP}={\rm co}\mathcal{NP},$$ contradicting Theorem \ref{theorem1}. This finishes the proof.\Q.E.D

\vskip 0.3 cm
\begin{remark}
\label{remark8.2}
Initially, propositional proof complexity was primarily concerned with proving lower bounds (even conditional ones) on the
lengths of proofs in propositional proof systems. This direction is extremely interesting and well justified in its own right, with the ultimate goal of settling whether $\mathcal{NP}={\rm co}\mathcal{NP}$ {\em (see e.g., \cite{CR79})}. In fact, as mentioned in Section \ref{sec:introduction}, such research directions are known as {\em Cook's Program for separating $\mathcal{NP}$ and ${\rm co}\mathcal{NP}$} {\em (see e.g., \cite{Bus12})}: prove superpolynomial lower bounds on proof lengths in stronger and stronger propositional proof systems until they are established for all abstract proof systems. Our approach, however, is ``to do the opposite": we first prove that $\mathcal{NP} \ne {\rm co}\mathcal{NP}$ and then apply this result to obtain lower bounds on the length of proofs in Frege proof systems.
\end{remark}

\vskip 0.3 cm
\section{Open Problems}
\label{sec:fundamental_open_problem}
\vskip 0.3 cm

Does $\mathcal{P}=\mathcal{NP}\cap{\rm co}\mathcal{NP}$? Currently, we have no answer to this problem. We have tried to resolve it but failed. We remark that this is a fascinating open problem in complexity theory. It is easy to see that $\mathcal{P}\subseteq\mathcal{NP}\cap{\rm co}\mathcal{NP}$, and the real problem is whether $\mathcal{NP}\cap{\rm co}\mathcal{NP}\subseteq\mathcal{P}$. Note that an example of a problem known to belong to both $\mathcal{NP}$ and ${\rm co}\mathcal{NP}$ (but not known to be in $\mathcal{P}$) is {\it Integer factorization} \cite{RSA78}: given positive integers $m$ and $n$, determine if $m$ has a factor less than $n$ and greater than one.

In addition, does the complexity class $\mathcal{BQP}$ (See page 213 of \cite{AB09} or, e.g., \cite{BV97}) contain the problem {\em Satisfiability} (i.e., can {\em Satisfiability} be solved by some bounded-error quantum polynomial-time Turing machines)? Note that by Corollary \ref{corollary1point4}, we have $\mathcal{BQP}\ne\mathcal{NP}$, but this does not exclude that $\mathcal{NP}\subsetneqq\mathcal{BQP}$. Thus, if one can show that {\em Satisfiability} or some other problem in $\mathcal{NP}$ is not in $\mathcal{BQP}$, then the precise relationships between $\mathcal{BQP}$ and $\mathcal{NP}$ would become clear.\footnote{It is widely believed that quantum polynomial-time Turing machines \cite{BV97} are unable to solve {\em Satisfiability} with an error probability of at most $\frac{1}{3}$ for all instances.}  Currently, we have no additional information about this problem. We have tried to resolve it but also failed, and we remark that this is likewise a fascinating open problem in complexity theory.

Finally, we hope that the methodology presented in this paper could shed light on other difficult complexity problems.

\vskip 0.3 cm
\bibliographystyle{aomplain}

\begin{thebibliography}{ABCDE11}

\bibitem[A1]{A1}
\bgroup\scshape{}Anonymous authors\egroup{}.
\newblock {\it Computational complexity theory}. Wikipedia, the free encyclopedia (Auguest, 2024).
\newblock Available at \href{https://en.wikipedia.org/wiki/Computational_complexity_theory}{/wiki/Computational complexity theory}

\bibitem[A2]{A2}
\bgroup\scshape{}Anonymous authors\egroup{}.
\newblock {\it NP (complexity)}. Wikipedia, the free encyclopedia (Auguest, 2024).
\newblock Available at \href{https://en.wikipedia.org/wiki/NP_(complexity)}{/wiki/NP (complexity)}

\bibitem[A3]{A3}
\bgroup\scshape{}Anonymous authors\egroup{}.
\newblock {\it Co-NP}. Wikipedia, the free encyclopedia (Auguest, 2024).
\newblock Available at \href{https://en.wikipedia.org/wiki/Co-NP}{/wiki/Co-NP}

\bibitem[A4]{A4}
\bgroup\scshape{}Anonymous authors\egroup{}.
\newblock {\it Proof theory}. Wikipedia, the free encyclopedia (Auguest, 2024).
\newblock Available at \href{https://en.wikipedia.org/wiki/Proof_theory}{/wiki/Proof theory}

\bibitem[A5]{A5}
\bgroup\scshape{}Anonymous authors\egroup{}.
\newblock {\it Proof complexity}. Wikipedia, the free encyclopedia (Auguest, 2024).
\newblock Available at \href{https://en.wikipedia.org/wiki/Proof_complexity}{/wiki/Proof complexity}

\bibitem[A6]{A6}
\bgroup\scshape{}Anonymous authors\egroup{}.
\newblock {\it Polynomial hierarchy}. Wikipedia, the free encyclopedia (March, 2025).
\newblock Available at \href{https://en.wikipedia.org/wiki/Polynomial_hierarchy}{/wiki/Polynomial hierarchy}

\bibitem[A7]{A7}
\bgroup\scshape{}Anonymous authors\egroup{}.
\newblock {\it Alan Turing}. Wikipedia, the free encyclopedia (April, 2025).
\newblock Available at \href{https://en.wikipedia.org/wiki/Alan_Turing}{/wiki/Alan Turing}

\bibitem[A8]{A8}
\bgroup\scshape{}Anonymous authors\egroup{}.
\newblock {\it Pairing function}. Wikipedia, the free encyclopedia (April, 2025).
\newblock Available at \href{https://en.wikipedia.org/wiki/Pairing_function}{/wiki/Pairing\_function}

\bibitem[AB09]{AB09}
\bgroup\scshape{}Sanjeev Arora and Boaz Barak\egroup{}.
\newblock {\it Computational Complexity: A Modern Approach.}
\newblock Cambridge University Press, 2009.

\bibitem[AHU74]{AHU74}
\bgroup\scshape{}Alfred V. Aho, John E. Hopcroft and Jeffrey D. Ullman\egroup{}.
\newblock{\it The Design and Analysis of Computer Algorithms.}
\newblock Addison--Wesley Publishing Company, Reading, California, 1974.

\bibitem[BGS75]{BGS75}
\bgroup\scshape{}Theodore Baker, John Gill, and Robert Solovay\egroup{}.
\newblock {\it Relativizations of The $P=?NP$ Question.}
\newblock SIAM Journal on Computing, Vol. 4, No. 4, December 1975, pp. 431--442. \href{https://doi.org/10.1137/0204037}{https://doi.org/10.1137/0204037}

\bibitem[BG98]{BG98}
\bgroup\scshape{}M. L. Bonet and N. Galesi\egroup{}.
\newblock{\it Linear Lower Bounds and Simulations in Frege Systems with Substitutions.}
\newblock In: M. Nielsen, W. Thomas (eds) Computer Science Logic. CSL 1997. Lecture Notes in Computer Science, vol. 1414, Springer, Berlin, Heidelberg. pp. 115--128, 1998.


\bibitem[BG04]{BG04}
\bgroup\scshape{}Ian F. Blake, and Theo Garefalakis\egroup{}.
\newblock {\it On the complexity of the discrete logarithm and Diffie--Hellman problems.}
\newblock Journal of Complexity 20 (2004) 148--170.
\href{https://doi.org/10.1016/j.jco.2004.01.002}{https://doi.org/10.1016/j.jco.2004.01.002}

\bibitem[BH19]{BH19}
\bgroup\scshape{}O. Beyersdorff and L. Hinde\egroup{}.
\newblock{\it Characterising tree-like Frege proofs for QBF.}
\newblock Information and Computation 268 (2019) 104429. \href{https://doi.org/10.1016/j.ic.2019.05.002}{https://doi.org/10.1016/j.ic.2019.05.002}

\bibitem[BP01]{BP01}
\bgroup\scshape{}P. Beame and T. Pitassi\egroup{}.
\newblock {\it Propositional proof complexity: past, present, and future.}
\newblock In: G. Paun, G. Rozenberg, A. Salomaa (Eds.), Current Trends in Theoretical Computer Science: Entering the $21$st Century, World Scientific Publishing, 2001, pp. 42--70.

\bibitem[BV97]{BV97}
\bgroup\scshape{}E. Bernstein and U. Vazirani\egroup{}.
\newblock{\it Quantum Complexity Theory}.
\newblock SIAM Journal on Computing, Vol. 26, No. 5, pp. 1411--1473, October 1997. \href{https://doi.org/10.1137/S0097539796300921}{https://doi.org/10.1137/S0097539796300921}.

\bibitem[Bus99]{Bus99}
\bgroup\scshape{}Samuel R. Buss\egroup{}.
\newblock {\it Propositional Proof Complexity: An Introduction.}
\newblock In: Computational Logic, edited by U. Berger and H. Schwichtenberg, Springer--Verlag, Berlin, 1999, pp. 127--178.

\bibitem[Bus12]{Bus12}
\bgroup\scshape{}Samuel R. Buss\egroup{}.
\newblock {\it Towards $NP-P$ via proof complexity and search.}
\newblock Annals of Pure and Applied Logic 163 (2012) 906--917. \href{https://doi.org/10.1016/j.apal.2011.09.009}{https://doi.org/10.1016/j.apal.2011.09.009}

\bibitem[Coo71]{Coo71}
\bgroup\scshape{}Stephen A. Cook\egroup{}.
\newblock {\it The complexity of theorem-proving procedures.}
\newblock In: Proceedings of the Third Annual ACM Symposium on Theory of Computing, May 1971, Pages 151--158. \href{https://doi.org/10.1145/800157.805047}{https://doi.org/10.1145/800157.805047}

\bibitem[CR74]{CR74}
\bgroup\scshape{}Stephen A. Cook and Robert A. Reckhow\egroup{}.
\newblock {\it On the lengths of proof in propositional calculus, preliminary version.}
\newblock In: Proceedings of the $6$th Annual ACM Symposium on the Theory of Computing, 1974, pp. 135--148. \href{https://doi.org/10.1145/800119.803893}{https://doi.org/10.1145/800119.803893}

\bibitem[CR79]{CR79}
\bgroup\scshape{}Stephen A. Cook and Robert A. Reckhow\egroup{}.
\newblock {\it The Relative Efficiency of Propositional Proof Systems.}
\newblock The Journal of Symbolic Logic, Volume 44, Number 1, March 1979,  pp. 36--50. 
\href{https://doi.org/10.2307/2273702}{https://doi.org/10.2307/2273702}

\bibitem[Coo00]{Coo00}
\bgroup\scshape{}Stephen A. Cook\egroup{}.
\newblock {\it The P versus NP problem.} In James Carlson, Arthur Jaffe, and Andrew Wiles (Eds.), {\it The Millennium Prize Problems} (pp. 87--104). American Mathematical Society. (Co-published with the Clay Mathematics Institute),  2006.
\newblock Also available at \href{http://www.cs.toronto.edu/~sacook/homepage/PvsNP.ps}{\fbox{this https URL}}.

\bibitem[CN10]{CN10}
\bgroup\scshape{}Stephen A. Cook and Phuong Nguyen\egroup{}.
\newblock {\it Logical Foundations of Proof Complexity.}
\newblock Association for Symbolic Logic, Cambridge University Press, 2010.

\bibitem[For00]{For00}
\bgroup\scshape{}Lance Fortnow\egroup{}.
\newblock {\it Diagonalization.}
\newblock Bulletin of the European Association for Theoretical Computer Science 71: 102--113 (2000).

\bibitem[For11]{For11}
\bgroup\scshape{}Lance Fortnow\egroup{}.
\newblock {\it A New Proof of the Nondeterministic Time Hierarchy.}
\newblock Computational Complexity Blog, 2011. Available at \href{https://blog.computationalcomplexity.org/2011/04/new-proof-of-nondeterministic-time.html}{this https URL.}

\bibitem[FS07]{FS07}
\bgroup\scshape{}Lance Fortnow and Rahul Santhanam\egroup{}.
\newblock {\it Time Hierarchies: A Survey.}
\newblock Electronic Colloquium on Computational Complexity, Report No. 4 (2007).

\bibitem[FS17]{FS17}
\bgroup\scshape{}Lance Fortnow and Rahul Santhanam\egroup{}.
\newblock {\it Robust Simulations and Significant Separations.}
\newblock Information and Computation 256 (2017) 149--159. \href{https://doi.org/10.1016/j.ic.2017.07.002}{https://doi.org/10.1016/j.ic.2017.07.002}. Also \href{https://arxiv.org/abs/1012.2034v1}{\color{blue}arXiv:1012.2034}

\bibitem[FSTW21]{FSTW21}
\bgroup\scshape{}Michael A. Forbes, Amir Shpilka, Iddo Tzameret and Avi Wigderson\egroup{}.
\newblock {\it Proof Complexity Lower Bounds from Algebraic Circuit Complexity.}
\newblock Theory of Computing, Volume 17 (10), 2021, pp. 1--88. \href{https://doi.org/10.4086/toc.2021.v017a010}{https://doi.org/10.4086/toc.2021.v017a010}

\bibitem[GJ79]{GJ79}
\bgroup\scshape{}Michael~R.~Garey and David~S.~Johnson\egroup{}.
\newblock {\it Computers and Intractability, a Guide to the Theory of NP-Completeness.}
\newblock W.~H. Freeman and Company, San Francisco, 1979.

\bibitem[Jer05]{Jer05}
\bgroup\scshape{}E. Jerabek\egroup{}.
\newblock {\it Weak Pigeonhole Principle, and Randomized Computation.}
\newblock Ph.D. thesis, Faculty of Mathematics and Physics, Charles University, Prague, 2005.

\bibitem[HS65]{HS65}
\bgroup\scshape{}J. Hartmanis and R. Stearns\egroup{}.
\newblock{\it On the computational complexity of algorithms.}
\newblock Transactions of the American Mathematical Society, 117: 285--306, 1965.

\bibitem[HCCRR93]{HCCRR93}
\bgroup\scshape{}Juris Hartmanis, Richard Chang, Suresh Chari, Desh Ranjan, and Pankaj Rohatgi\egroup{}.
\newblock {\it Relativization: a Revisionistic Retrospective.}
\newblock Current Trends in Theoretical Computer Science, 1993, pp. 537--547.
\href{https://doi.org/10.1142/9789812794499_0040}{https://doi.org/10.1142/9789812794499\_0040}

\bibitem[HU79]{HU79}
John E. Hopcroft and Jeffrey D. Ullman.
\newblock{\it Introduction to Automata Theory, Languages, and Computation.}
\newblock Addison--Wesley Publishing Company, 1979.

\bibitem[Hop84]{Hop84}
\bgroup\scshape{}John. E. Hopcroft\egroup{}.
\newblock {\it Turing machines.}
\newblock Scientific American, May 1984, pp. 86--98.

\bibitem[Kar72]{Kar72}
\bgroup\scshape{}Richard M. Karp\egroup{}.
\newblock {\it Reducibility among combinatorial problems.}
\newblock In: Miller R. E., Thatcher J. W., Bohlinger J. D. (eds) Complexity of Computer Computations., Plenum Press, New York, 1972, 85--103. \href{https://doi.org/10.1007/978-1-4684-2001-2_9}{https://doi.org/10.1007/978-1-4684-2001-2\_9}

\bibitem[Kra95]{Kra95}
\bgroup\scshape{}J. Krajicek\egroup{}.
\newblock {\it Bounded Arithmetic, Propositional Logic, and Complexity Theory.}
\newblock Cambridge, 1995.

\bibitem[Kra19]{Kra19}
\bgroup\scshape{}J. Krajicek\egroup{}.
\newblock {\it Proof complexity.}
\newblock Encyclopedia of Mathematics and Its Applications, Vol. 170, Cambridge University Press, 2019.

\bibitem[Lad75]{Lad75}
\bgroup\scshape{}Richard E. Ladner\egroup{}.
\newblock {\it On the Structure of Polynomial Time Reducibility.}
\newblock Journal of the ACM, Vol. 22, No. 1, January 1975, pp. 155-171. \href{https://doi.org/10.1145/321864.321877}{https://doi.org/10.1145/321864.321877}

\bibitem[Lev73]{Lev73}
\bgroup\scshape{}Leonid A. Levin\egroup{}.
\newblock {\it Universal search problems} (in Russian).
\newblock Problemy Peredachi Informatsii 9 (1973), 265--266. English translation in B.~A.~Trakhtenbrot, {\it A survey of Russian approaches to Perebor (brute-force search) algorithms}, Annals of the History of Computing 6 (1984), 384--400.

\bibitem[Law85]{Law85}
\bgroup\scshape{}E. L. Lawler\egroup{}.
\newblock {\it The Traveling salesman problem: a guided tour of combinatorial optimization.}
\newblock John Wiley \& Sons, 1985.

\bibitem[Lin21]{Lin21}
\bgroup\scshape{}Tianrong Lin\egroup{}.
\newblock {\it Diagonalization of Polynomial--Time Deterministic Turing Machines via Nondeterministic Turing Machines.}
\newblock arXiv: 2110.06211, 2021. Available at \href{https://arxiv.org/abs/2110.06211}{arXiv: 2110.06211}

\bibitem[MS72]{MS72}
\bgroup\scshape{}Albert R. Meyer and Larry J. Stockmeyer\egroup{}.
\newblock{\it The equivalence problem for regular expressions with squaring requires exponential space.}
\newblock In: IEEE 13th Annual Symposium on Switching and Automata Theory (swat 1972), pp. 125--129. \href{https://doi.org/10.1109/SWAT.1972.29}{https://doi.org/10.1109/SWAT.1972.29}

\bibitem[Pap94]{Pap94}
\bgroup\scshape{}Christos H. Papadimitriou\egroup{}.
\newblock {\it Computational Complexity.}
\newblock Addison--Wesley, 1994.

\bibitem[Pud08]{Pud08}
\bgroup\scshape{}Pavel Pudlák\egroup{}.
\newblock {\it Twelve Problems in Proof Complexity.}
\newblock In: E.A. Hirsch et al. (Eds.): CSR 2008. Lecture Notes in Computer Science, vol. 5010, Springer, Berlin, Heidelberg, pp. 13--27, 2008. 

\bibitem[Raz03]{Raz03}
\bgroup\scshape{}Alexander A. Razborov\egroup{}.
\newblock{\it Propositional Proof Complexity.}
\newblock Journal of the ACM, Vol. 50, No. 1, January 2003, pp. 80--82. \href{https://doi.org/10.1145/602382.602406}{https://doi.org/10.1145/602382.602406}

\bibitem[Raz04]{Raz04}
\bgroup\scshape{}Alexander A. Razborov\egroup{}.
\newblock{\it Feasible proofs and computations: partnership and fusion.}
\newblock In: J. Diaz, J. Karhumaki, A. Lepisto, D. Sannella (eds) Automata, Langugages and Programming. ICALP 2004. Lecture Notes in Computer Science, vol. 3142, Springer, Berlin, Heidelberg, 2004, pp. 8--14. 

\bibitem[Raz15]{Raz15}
\bgroup\scshape{}Alexander A. Razborov\egroup{}.
\newblock{\it Pseudorandom generators hard for $k$-DNF resolution and polynomial calculus resolution.}
\newblock Annals of Mathematics 181 (2015), 415--472. \href{https://doi.org/10.4007/annals.2015.181.2.1}{https://doi.org/10.4007/annals.2015.181.2.1}

\bibitem[Rec76]{Rec76}
\bgroup\scshape{}Robert A. Reckhow\egroup{}.
\newblock {\it On the lengths of proofs in the propositional calculus.}
\newblock Ph.D. Thesis, Department of Computer Science, University of Toronto, 1976.

\bibitem[RSA78]{RSA78}
\bgroup\scshape{}R. L. Rivest, A. Shamir, and L. Adleman\egroup.
\newblock {\it A Method for Obtaining Digital Signatures and Public-Key Cryptosystems.}
\newblock Communications of the Association for Computing Machinery, Volume 21, Issue 2, February 1978, Pages 120--126.
\href{https://doi.org/10.1145/359340.359342}{https://doi.org/10.1145/359340.359342}

\bibitem[Sip13]{Sip13}
\bgroup\scshape{}Michael Sipser\egroup{}.
\newblock {\it Introduction to the theory of computation.}
\newblock Third Edition, Cengage Learnng, 2013.

\bibitem[Sto77]{Sto77}
\bgroup\scshape{}Larry J. Stockmeyer\egroup{}.
\newblock{\it The polynomial-time hierarchy.}
\newblock Theoretical Computer Science 3 (1977) 1--22.  \href{https://doi.org/10.1016/0304-3975(76)90061-X}{https://doi.org/10.1016/0304-3975(76)90061-X}

\bibitem[SFM78]{SFM78}
\bgroup\scshape{}Joel I. Seiferas, Michael J. Fischer and Albert R. Meyer\egroup{}.
\newblock {\it Separating Nondeterministic Time Complexity Classes.}
\newblock Journal of the ACM, Vol. 25, No. 1, January 1978, pp. 146--167. 
\href{https://doi.org/10.1145/322047.322061}{https://doi.org/10.1145/322047.322061}

\bibitem[Tao22]{Tao22}
\bgroup\scshape{}Terence Tao\egroup{}.
\newblock{\it Analysis I,} Texts and Readings in Mathematics 37, (Fourth Edition).
\newblock  Hindustan Book Agency, Springer, 2022.
\href{https://doi.org/10.1007/978-981-19-7261-4}{https://doi.org/10.1007/978-981-19-7261-4}

\bibitem[Tur37]{Tur37}
\bgroup\scshape{}Alan M. Turing\egroup{}.
\newblock{\it On computable numbers with an application to the entscheidnungsproblem.}
\newblock Proceedings of the London Mathematical Society, Volume s2-42, Issue 1, 1937, Pages 230--265.
\href{https://doi.org/10.1016/0066-4138(60)90045-8}{https://doi.org/10.1016/0066-4138(60)90045-8}

\bibitem[Wig07]{Wig07}
\bgroup\scshape{}Avi Wigderson\egroup{}.
\newblock {\it $P$, $NP$ and Mathematics--a computational complexity perspective.}
\newblock Proceedings of the ICM 06, vol. 1, EMS Publishing House, Zurich, pp. 665--712, 2007.

\bibitem[Yao85]{Yao85}
\bgroup\scshape{}Andrew Chi-Chih Yao\egroup{}.
\newblock {\it Separating the Polynomial-Time Hierarchy by Oracles.}
\newblock In: Proceedings of the 26th Annual IEEE Symposium on Foundations of Computer Science, pp. 1--10, 1985. \href{https://doi.org/10.1109/SFCS.1985.49}{https://doi.org/10.1109/SFCS.1985.49}

\bibitem[\v{Z}\'{a}k83]{Zak83}
\bgroup\scshape{}S. \v{Z}\'{a}k\egroup{}.
\newblock {\it A Turing machine time hierarchy.}
\newblock Theoretical Computer Science 26 (1983) 327--333.
\href{https://doi.org/10.1016/0304-3975(83)90015-4}{https://doi.org/10.1016/0304-3975(83)90015-4}.

\end{thebibliography}

%\clearpage
\appendix

\vskip 0.3 cm
\section{Two Definitions of $\mathcal{NP}$ Are Equivalent}
\label{sec:appendix}
\vskip 0.3 cm

The official description of the definition of $\mathcal{NP}$ in \cite{Coo00} is given by the following.
\begin{definition}\label{defNP}
   $$
   \mathcal{NP}\overset{{\rm def}}{=}\{L\,:\,w\in L\Leftrightarrow\exists y(|y|\leq|w|^k\,\,\,\text{and } R(w,y))\},
   $$
where the {\it checking relation} $$R(w,y)$$ is polynomial-time, i.e., the language $$L_R\overset{{\rm def}}{=}\{w\#y\,:\,R(w,y)\}$$ is in $\mathcal{P}$.
\end{definition}

The complexity class $\mathcal{P}$, which is referred to in Definition \ref{defNP}, is defined as follows: 
\begin{definition}
\label{defP}
   $$
      \begin{array}{ll}
        \mathcal{P}\overset{{\rm def}}{=}&\{L\,:\,L=L(M)\text{ for some Turing machine $M$ that runs} \\
         &\qquad\text{ in polynomial time }\}.
      \end{array}
   $$
\end{definition}
   
In Definition \ref{defP}, the Turing machine $M$ is understood to be deterministic, and the notion of ``running in polynomial time" means that there exists a $k$ ($\in\mathbb{N}_1$) such that, $$T_M(n)\leq n^k+k,\quad\forall n\in\mathbb{N}_1, $$ where $T_M(n)$ is defined by $$T_M(n)=\max\{t_M(w)\,:\,w\in\Sigma^n\}$$ and $t_M(w)$ denotes the number of steps (or moves) in the computation of $M$ on input $w$ of length $n$.

As the reader may note, this definition (i.e., Definition \ref{defNP}) looks somewhat different from the one given in Section \ref{sec:preliminaries}. Nevertheless, the two definitions are equivalent. Here we show the equivalence of these two definitions of $\mathcal{NP}$.
  
\vskip 0.3 cm
\subsection{A Proof that Two Definitions of $\mathcal{NP}$ Are Equivalent}
\vskip 0.3 cm
  
The remainder of this appendix is devoted to showing that {\em Definition} \ref{defNP} and the definition of $\mathcal{NP}$ given in Section \ref{sec:preliminaries} are equivalent.

To begin, we first give another definition of $\mathcal{NP}$, which is similar in form to the official description of the definition of $\mathcal{P}$ in \cite{Coo00} (i.e., Definition \ref{defP} above):
\begin{definition}   
\label{defaNP}
$$
  \begin{array}{ll}
     \mathcal{NP}\overset{{\rm def}}{=}&\{\mathcal{L}\,:\,\mathcal{L}=L(M)\text{ for some nondeterministic Turing machine $M$} \\
     &\qquad\text{that runs in polynomial time}\}.
  \end{array}
$$
\end{definition}
   
We now complete the last step of the task mentioned at the beginning of this subsection.
\vskip 0.2 cm
\begin{proof}
We need to show that Definition \ref{defaNP} is equivalent to Definition \ref{defNP} and that Definition \ref{defaNP} is equivalent to the definition of $\mathcal{NP}$ given in Section \ref{sec:preliminaries}. 

By a simple argument it is clear that Definition \ref{defaNP} and the definition of $\mathcal{NP}$ given in Section \ref{sec:preliminaries} are equivalent. We proceed as follows:

We first show the ``if" direction. Suppose that $$L\in\bigcup_{i\in\mathbb{N}_1}{\rm NTIME}[n^i]. $$ Then there exists a $k\in\mathbb{N}_1$ such that $$L\in{\rm NTIME}[n^k], $$ which means that for all $n\in\mathbb{N}$, there is a nondeterministic Turing machine $M$, for any $w\in\Sigma^n $ $$T_M(|w|)\leq c_0n^k+c_1n^{k-1}+\cdots +c_{k-1}n+c_k\quad\text{(where $c_0>0$)}, $$ and $L=L(M)$.

For such constants $c_0$, $c_1$, $\cdots$, $c_k$, there must exist a minimal $t\in\mathbb{N}_1$ such that, for all $n\in\mathbb{N}$ and for any $w\in\Sigma^n,$  $$T_M(|w|)\leq n^t+t. $$Such a $t$ is easy to find; for example, $t=k+1$ may works. Indeed, for sufficiently large positive integer $n$,
$$\aligned
T_M(|w|)\,\leq\, &c_0n^k+c_1n^{k-1}+\cdots +c_{k-1}n+c_k\\
\,\leq\, & c_0n^k+c_1n^k+\cdots +c_{k-1}n^k+c_kn^k\\
\,=\,&\left(\sum_{i=0}^kc_i\right)n^k\\
\,<\,&n^{k+1}+(k+1).
\endaligned$$ Consequently, $$L\in\{\mathcal{L}\,:\,\mathcal{L}=L(M)\text{ for some nondeterministic Turing machine $M$ that runs in polynomial time}\}. $$

\vskip 0.3 cm

We next show the ``only if" direction. Suppose now that $$L\in \{\mathcal{L}\,:\,\mathcal{L}=L(M)\,\text{ for some nondeterministic Turing machine $M$ that runs in polynomial time}\}. $$Then there exists a $k\in\mathbb{N}_1$ such that, for all $n\in\mathbb{N}$ and for all $w\in\Sigma^n$, $$T_M(|w|)\leq n^k+k,$$ which implies that $$L\in{\rm NTIME}[n^k]\subseteq\bigcup_{i\in\mathbb{N}_1}{\rm NTIME}[n^i]. $$

It remains to show that Definition \ref{defaNP} is equivalent to Definition \ref{defNP}. For the sake of brevity we omit the proof and refer the reader to \cite{Kar72} (cf. Theorem 1 in \cite{Kar72}) or to the proof of Theorem $7$.$20$ in \cite{Sip13}, which answers the same question discussed here (see p. $294$ of \cite{Sip13}). Therefore, this completes the proof.
\end{proof}

\section{A New Construction of $L_s^{'i}$}
\label{sec:a_new_construction_of_lsi}

Note that the language $L_s^{'i}$ is accepted by $U$ while running in at most $n^{i+1}$ steps (which means that $L_s^{'i}\in\mathcal{NP}$). This means that $U$ simulates those $n^k+k$ time-bounded nondeterministic Turing machines of order $k$ for which $k\leq i$, since the simulation takes $$(n^k+k)\log(n^k+k)< n^{k+1}$$ steps, which is required to be at most $n^{i+1}$ steps (all such $k$ obviously satisfy $k\leq i$, because this requires $k+1\leq i+1$). Thus, for a fixed $i\in\mathbb{N}_1$, the language $L_s^{'i}$ can be constructed as follows:

Since all ${\rm co}\mathcal{NP}$ machines are enumerable, let $U$ be a four-tape nondeterministic Turing machine that operates as follows on an input string $x$ of length $n$.
\begin{enumerate}
\item{
$U$ decodes the tuple encoded by $x$. If $x$ is not the encoding of some ${\rm co}\mathcal{NP}$ machine $M_l$ for some $l$ (by Definition \ref{definition2.4}, equivalently, if $x$ cannot be decoded into some polynomial-time nondeterministic Turing machine, i.e., if $x$ represents the trivial Turing machine with an empty next-move function), then GOTO $5$; else determine $t$, the number of tape symbols used by $M_l$; $s$, its number of states; and $k$, its order. The third tape of $U$ can be used as ``scratch" memory to calculate $t$.
}
\item{
If $k+1\leq i+1$, then GOTO $3$; else reject the input $x$.\label{item2}
}
\item{
$U$ then lays off on its second tape $n$ blocks of $\lceil\log t\rceil$ cells each, the blocks being separated by a single cell holding a marker $\#$, i.e., there are $(1+\lceil\log t\rceil)n$ cells in all. Each tape symbol occurring in a cell of $M_l$'s tape will be encoded as a binary number in the corresponding block of the second tape of $U$. Initially, $U$ places $M_l$'s input, in binary-coded form, in the blocks of tape $2$, filling the unused blocks with the code for the blank.
}
\item{
On tape $3$, $U$ sets up a block of $\lceil(k+1)\log n\rceil$ cells, initialized to all $0$'s. Tape $3$ is used as a counter to count up to $n^{k+1}$.
}
\item{
Using nondeterminism, $U$ simulates $M_l$, using tape $1$ (its input tape) to determine the moves of $M_l$ and tape $2$ to simulate the tape of $M_l$. The moves of $U$ are counted in binary in the block of tape $3$, and tape $4$ is used to hold the states of $M_l$. If the counter on tape $3$ overflows, $U$ halts without accepting. The specified simulation is the following:

 ``On input $x=1^{n-m}\langle M_l\rangle$ where $m=|\langle M_l\rangle|$, $U$ on input $1^{n-m}\langle M_l\rangle$ simulates $M_l$ on input $1^{n-m}\langle M_l\rangle$ using nondeterminism in $$ n^{k+1} $$ time and outputs its answer, i.e., accepting if $M_l$ accepts and rejecting if $M_l$ rejects." 

\vskip 0.2cm
It should be pointed out that, although we say that input $x$ is decoded into a ${\rm co}\mathcal{NP}$ machine $M_l$, $U$ in fact simulates the $n^k+k$ time-bounded nondeterministic Turing machine $M_l$. That is to say, $U$ accepts the input $x=1^{n-m}\langle M_l\rangle$ if there exists at least one accepting path of $M_l$ on input $x$, and $U$ rejects the input $x=1^{n-m}\langle M_l\rangle$ if there exists no accepting path of $M_l$ on input $x$.
}
\item{
Since $x$ is an encoding of the trivial Turing machine with an empty next-move function, $U$ sets up a block of 

$$\lceil 2\times\log n\rceil$$ cells on tape $3$, initialized to all $0$'s. Tape $3$ is used as a counter to count up to $$n^2.$$

Using its nondeterministic choices, $U$ moves according to the path given by $x$. The moves of $U$ are counted in binary in the block of tape $3$. If the counter on tape $3$ overflows, then $U$ halts. $U$ accepts $x$ if and only if there is a computation path from the start state of $U$ leading to the accept state and the total number of moves does not exceed $$n^2$$ steps, so the computation lies within $O(n)$ steps. Note that the factor $2$ in $$\lceil 2\times\log n\rceil$$ is fixed; it is a default setting.
}
\end{enumerate}

Compared with the proof of Theorem \ref{theorem4.1}, the construction of $U$ differs only in item \ref{item2}.

It is easy to see that $L_s^{'i}\subseteq L_s^i$ for some $i$, since for any $w\in L_s^{'i}$ we have $w\in L_s^i$, but for some $w\in L_s^i$, it may be that $w\notin L_s^{'i}$ (we leave it to the reader to determine why). Nevertheless, since all polynomial-time nondeterministic Turing machines are simulated by $U$ as $i$ runs through all numbers in $\mathbb{N}_1$, we certainly have $$L_s=\bigcup_{i\in\mathbb{N}_1}L_s^{'i}. $$

\vskip 0.5 cm
\end{document}